\documentclass[11pt]{article}
\usepackage[a4paper, portrait, margin=1in]{geometry}
\usepackage[utf8]{inputenc}
\usepackage[numbers, sort, compress]{natbib}
\usepackage{nicefrac}
\usepackage{multirow}
\usepackage{soul}
\usepackage[dvipsnames,table]{xcolor}

\usepackage{float}

\usepackage{amsmath, mathtools, amsthm,thmtools}

\usepackage{nicefrac}

\usepackage{amsfonts}
\usepackage{amssymb}

\usepackage{hyperref}
\hypersetup{
    colorlinks=true,
    linkcolor=blue,
    filecolor=blue,      
    urlcolor=blue,
    citecolor=blue,
    pdftitle={Triangle_RFDS_Multioccur},
    pdfpagemode=FullScreen,
    }

\usepackage{cleveref}
\usepackage{algorithm, times}
\usepackage[noend]{algpseudocode}

\theoremstyle{plain}
\newtheorem{theorem}{Theorem}
\newtheorem{lemma}[theorem]{Lemma}
\newtheorem{claim}{Claim}
\newtheorem{corollary}[theorem]{Corollary}

\theoremstyle{definition}
\newtheorem{definition}[theorem]{Definition}

\theoremstyle{remark}
\newtheorem{remark}[theorem]{Remark}

\usepackage[most]{tcolorbox}
\newtcolorbox{idea}[1][]
{
colbacktitle=cyan,
colback=cyan!10,
arc=1pt,
boxrule=1pt,
title=#1 
}

\newtcolorbox{update}[1][]
{
colbacktitle=gray,
colback=gray!10,
arc=1pt,
boxrule=1pt,
title=#1 
}

\newtcolorbox{question}[1][]
{
coltitle=black,
colbacktitle=yellow,
colback=yellow!10,
arc=1pt,
boxrule=1pt,
title=#1 
}

\newtcolorbox{note}[1][]
{
coltitle=black,
colbacktitle=green,
colback=green!10,
arc=1pt,
boxrule=1pt,
title=#1 
}

\newtcolorbox{problem}[1][]
{
coltitle=black,
colbacktitle=red!60,
colback=red!10,
arc=1pt,
boxrule=1pt,
title=#1 
}

\newcommand{\bs}{\boldsymbol{s}}

\newcommand{\bF}{\boldsymbol{F}}

\newcommand{\bS}{\boldsymbol{S}}

\newcommand{\sA}{\mathcal{A}}

\newcommand{\sG}{\mathcal{G}}

\newcommand{\sU}{\mathcal{U}}

\newcommand{\numtriangle}{T}
\newcommand{\numtlb}{\overline{T}}
\newcommand{\esttriangle}{\widetilde{T}}

\newcommand{\triangleset}{\Delta}
\renewcommand{\triangle}{\tau} 

\newcommand{\samplededges}{S_\edgeset}
\newcommand{\sampledtriangles}{S_\numtriangle}

\newcommand{\vertexsensitivity}{\Delta_\vertexset}
\newcommand{\edgesensitivity}{\Delta_\edgeset}
\newcommand{\vertexhash}{h_\vertex}
\newcommand{\edgehash}{h_\edge}
\newcommand{\tsamprate}{r}
\newcommand{\vsamprate}{p}

\newcommand{\esamprate}{q}
\newcommand{\vsenseub}{\overline{\Delta}_\vertexset}
\newcommand{\esenseub}{\overline{\Delta}_\edgeset}

\newcommand{\appcon}{\fbrac{\approxerror,\confidence}}

\newcommand{\numcliques}[1]{{n_{#1}}}

\newcommand{\estcliques}[1]{\widetilde{\numcliques{#1}}}

\newcommand{\cKMV}{\texttt{Clipped{-}KMV}}
\newcommand{\KMV}{\texttt{KMV}}

\newcommand{\field}[1]{\mathbb{#1}}

\newcommand{\func}[3]{{#1} : {#2} \rightarrow {#3}}
\newcommand{\ceil}[1]{\left\lceil{#1}\right\rceil}
\newcommand{\floor}[1]{\left\lfloor{#1}\right\rfloor}

\newcommand{\sequence}[2]{{#1}_1,{#1}_2,...,{#1}_{#2}}
\newcommand{\fbrac}[1]{\left({#1}\right)}
\newcommand{\sbrac}[1]{\left\{{#1}\right\}}
\newcommand{\tbrac}[1]{\left[{#1}\right]}
\newcommand{\abs}[1]{\left|{#1}\right|}
\newcommand{\size}[1]{\left|{#1}\right|}

\newcommand{\poly}[1]{\mathrm{poly}\fbrac{#1}}

\DeclareMathOperator*{\E}{\field{E}}
\DeclareMathOperator*{\Var}{\mathrm{Var}}

\newcommand{\approxerror}{\varepsilon}
\newcommand{\confidence}{\delta}

\newcommand{\bigo}[1]{O\fbrac{{#1}}}
\newcommand{\bigot}[1]{\widetilde{O}\fbrac{{#1}}}
\newcommand{\smallo}[1]{o\fbrac{{#1}}}

\newcommand{\bigomega}[1]{\Omega\fbrac{{#1}}}

\newcommand{\graph}{G}
\newcommand{\vertexset}{V}
\newcommand{\edgeset}{E}

\newcommand{\vertexcount}{n}
\newcommand{\edgecount}{m}
\newcommand{\vertex}{v}
\newcommand{\altvertex}{u}
\newcommand{\edge}{e}

\usepackage{lineno}
\modulolinenumbers[5]
\usepackage{todonotes}

\title{Counting Triangles in Graph Streams with Repeatable and Forgettable Edges}

\author{
  \begin{tabular}{cc}
    Sourav Chakraborty & Debarshi Chanda \\
    \textit{Indian Statistical Institute} & \textit{Indian Statistical Institute} \\
     & \\
     Arijit Ghosh & A. Pavan\\
     \textit{Indian Statistical Institute} & \textit{Iowa State University}\\
     &\\
     Chhaya Trehan & N.~V.~Vinodchandran\\
     \textit{Reykjavik University} & \textit{University of Nebraska--Lincoln}
  \end{tabular}
}

\date{}

\begin{document}

\maketitle

\thispagestyle{empty}
\begin{abstract}
Most existing graph streaming algorithms assume the ideal scenario where each edge arrives only once. Real-world graph streams, such as communication or transaction logs, often contain many repeated occurrences of the same edge. In general, the algorithms developed for the single-edge arrival case can fail when edges can arrive multiple times. Motivated by this, we study the {\em repeated-edge arrival graph streaming model} where an edge is allowed to arrive multiple times.

In this work, we study the triangle counting problem in the repeated-edge arrival model: approximate the number of triangles in the underlying {\em simple graph} despite arbitrary edge repetitions. 
We design the first algorithms for triangle counting with optimal\footnote{Optimal up to polylog factors.} space complexity.  
In particular, we present a single-pass algorithm that computes an $(\varepsilon,\delta)$-approximation of the number of triangles with optimal space complexity. We introduce {\em right-to-be-forgotten graph streaming} (RFGS) model, where a forget
operation can cause all previous occurrences of an edge to disappear. We show that our single-pass algorithm can be extended to the RFGS model with optimal space complexity. Finally, we present optimal constant-pass algorithms that compute an $(\varepsilon,\delta)$-approximation of the number of triangles and cliques for the repeated-edge arrival graph streams. 
\end{abstract}

\newpage
\thispagestyle{empty}
\tableofcontents

\newpage
\setcounter{page}{1}

\maketitle


\section{Introduction}








Counting triangles in large graphs is a fundamental task in network analysis. Triangles capture the presence of tightly knit groups and form the basis of widely used structural measures such as transitivity and clustering coefficients~\cite{Newman2003,Wasserman94}. They also play a central role in applications ranging from community detection and link prediction to fraud and spam detection in communication and online social networks ~\cite{BBCG08,Milo2002NetworkMotifs,Ugander2013SubgraphFrequencies}. In many modern environments, graph data arrives as a rapid stream of edges (e.g., communication logs, social interactions, or web activity), making it infeasible to store the entire graph. This motivates the study of triangle counting within the streaming model, where the edges of the graph arrive as a stream, and the algorithm must process edges as they arrive. In the design of streaming algorithms, the space complexity is a primary concern.

A substantial body of work has addressed the problem of estimating the number of triangles in graph streams. Early work by Bar-Yossef et al.\ \cite{BarYossef2002Reductions} introduced streaming formulations of the problem and developed approximation algorithms via reductions to frequency-moment estimation. Since then, the problem has attracted significant attention from both theoretical and practical perspectives. In particular, researchers have investigated triangle counting in both single-pass and multi-pass streaming models, with single-pass getting a majority of attention. The state-of-the-art space-optimal single-pass algorithm is due to Jayaram and Kallaugher~\cite{JK21}, whose space complexity matches the lower bounds from Braverman et al.~\cite{Braverman2013HowHard} and Kallaugher and Price~\cite{KallaugherP17}. In the multipass setting, the optimal algorithms are due to 
McGregor et al.~\cite{McGregorVV/pods/2016/BetterTriangleCountingStreams}, and Bera and Chakrabarti \cite{BeraC/stacs/2017/CliqueCountingEA}.


Much of the work on graph streams assumes that the underlying streaming graph is simple, i.e., each edge arrives only once. We call this setting {\em single-edge arrival} model. However, in many real-world graph streams, edges do not arrive as unique events; instead, the same edge may appear multiple times due to repeated interactions, duplicated logging, or distributed/parallel data-collection mechanisms. Communication logs repeatedly record messages between the same pair of users; financial transaction streams contain many transfers between the same accounts.  It is known that many real-world graphs from the SNAP database~\cite{snapnets} are indeed multi-graphs.  
In a temporal network, each interaction between two vertices is represented by a timestamped edge event, and multiple interactions between the same pair produce multiple edge arrivals. Thus, any temporal graph stream is inherently a multigraph stream. For example, the CollegeMsg Temporal network from SNAP database, has more than $59K$ edges, but its underlying simple graph only has around $20K$ edges. 


The present work investigates the {\em repeated-edge arrival} model, a streaming setting in which the same edge may appear multiple times in the input stream. Even though the streaming algorithm observes a multi-graph, its goal is to estimate the statistics of interest in the underlying {\em simple graph}.  This model has been considered previously in the literature. Notably, the works of~\cite{JPS15, WQSZTG17, XFL25} examine the problem of estimating the number of triangles in such graph streams. These papers highlight the statistical bias introduced by repeated arrivals, specifically, how naively treating each arrival as a distinct edge can distort estimators—and they propose techniques to correct for this bias. However, the focus of the prior studies is on empirical performance rather than deriving optimal space bounds. A theoretical understanding for repeated-edge arrival streams remains only partially understood, and significant gaps remain in the space complexity of fundamental graph problems in this setting. 


The Right to be Forgotten Data Stream model (RFDS model) was introduced 
in~\cite{PavanCVM/PODS/2024/RFDSIntro} to capture streaming computation in settings where previously observed data may later be forgotten. Much of the data collected in modern systems is inherently relational: social networks record friendships and interactions, communication networks record exchanges between users, and financial networks record transactions between accounts. In such settings, privacy or regulatory requirements to delete previously collected data naturally translate into requests to forget edges of the underlying graph. This motivates the {\em Right to be Forgotten Graph Streaming Model} (RFGS) graph streaming model with \emph{forget} operations, where a request to forget an edge removes all prior occurrences of the edge from the stream. Forget operations pose new challenges in the design of graph streaming algorithms, as a forgotten edge may have already contributed to many substructures in the graph.

In this work, we design optimal algorithms for triangle counting in the repeated edge arrival model in settings including single-pass, multi-pass, and establish the first triangle counting algorithm with optimal space complexity in the RFGS model.

\subsection{Our Results}



We design space-optimal algorithms for estimating the number of triangles in repeated-edge arrival and RFGS models. Our main contribution is that in many scenarios, it is possible to design algorithms that meet the bounds of the algorithms in the single-edge arrival model. In particular, we design near-optimal single-pass and multi-pass algorithms for approximating the number of triangles across both models. Below is a summary of our results.  
\begin{enumerate}
\item Our first result is an optimal single-pass algorithm for apprximating the number of triangles in the repeated-edge arrival model. This 
algorithm takes $\widetilde{O}\left({\frac{m}{\varepsilon^2T}}(\edgesensitivity + \sqrt{\vertexsensitivity})\right)$ space to output an $\varepsilon$ approximation. Here $m$ is the number of edges, $T$ is the number of triangles, $\edgesensitivity$ (respectively, $\vertexsensitivity$) is the maximum number of triangles that are incident on an edge (respectively, vertex)~\cite{JK21}. 

\item We introduce the {\em Right to be Forgotten Graph Stream} (RFGS) model. This is an extension of the RFDS (Right to be Forgotten in Data Streams) model, studied in~\cite{PavanCVM/PODS/2024/RFDSIntro,LNSW25}. In {the RFGS} model each element in the stream is an update of the form $(\edge,U)$, where $U \in \sbrac{+,\perp}$ is the update information: an $+$ update corresponds to insertion, and $\perp$ update corresponds to {\em forget} upon which all the prior occurrences of edge $e$ in the stream are removed. Building on the ideas from the repeated-edge arrival model algorithm, we design an optimal algorithm to estimate the number of triangles in the RFGS model.
The space complexity of our algorithm is $\bigot{\fbrac{\edgecount_{max}\fbrac{\frac{\edgesensitivity}{\numtlb}+\frac{\sqrt{\vertexsensitivity}}{\numtlb}}+\frac{\numtriangle_{\max{}}}{\numtriangle}}}$, where $m_{max}$ and $T_{max}$ are the maximum number of edges and triangles present in the graph at any time. We establish a matching lower bound on the space complexity. 



\item Finally, we consider multi-pass algorithms for repeated-edge arrival model, where the algorithm is allowed to read the stream more than once. 
Our first algorithm is a 2-pass algorithm that uses $\tilde{O}\left(\frac{m}{\varepsilon^2 \sqrt{T}}\right)$ space and is 
suitable for input graphs with fewer than $m$ triangles ($\numtriangle \le m$). Our second algorithm is a 4-pass algorithm that uses space $\tilde{O}\left(\frac{m^{3/2}}{\varepsilon^2 T}\right)$ and  
is suitable for graphs with more than $m$ triangles ($\numtriangle > m$). These are optimal in the constant pass setting as \cite{BeraC/stacs/2017/CliqueCountingEA} provides a $\bigomega{\min\fbrac{\frac{\edgecount}{\sqrt{\numtriangle}},\frac{\edgecount^{3/2}}{\numtriangle}}}$ lower bound for any constant pass algorithm for triangle counting in the single-edge arrival model. We can also extend these algorithms to also estimate the number of $k$-cliques. 

\end{enumerate}


\subsection{Related Work}
Triangle counting in the data streaming model has been extensively studied over the last two decades, motivated by applications in massive graph analytics, network science, and subgraph estimation. Early work of Bar-Yossef et al.~\cite{BarYossef2002Reductions} connected triangle counting to the estimation of frequency moments in data streams and established some of the first streaming formulations for the problem. Subsequently, Buriol et al.~\cite{Buriol2006CountingTriangles} introduced sampling-based algorithms to estimate the number of triangles. Following this, there has been an extensive body of work this problem. In particular ~\cite{Jowhari2005NewStreaming,Pavan2013CountingSampling,JPS15, WQSZTG17,XFL25,Cormode2014SecondLook,McGregorVV/pods/2016/BetterTriangleCountingStreams, McGregorVorotnikova2020,StefaniERU16,TsourakakisPagh11,JowhariR26} is a non-exhaustive list of works on triangle counting in graph data streams. All these works consider the single-edge arrival model. 


Some known triangle estimation algorithms, developed for the single-edge arrival model, can be extended to the repeated-edge arrival model with appropriate modifications: Specifically, the algorithms that create a subgraph by sampling the edges and scaling the number of triangles in the sampled subgraph (for example~\cite{StefaniERU16, TsourakakisPagh11}. Such algorithms can be modified by using hash function-based sampling or $\ell_0$-sampling to deal with repeated edge occurrences. However, many other algorithms do not seem easily extendable to the repeated-edge arrival model. Triangle counting has also been studied in the multi-pass streaming model~\cite{Braverman2013HowHard, Cormode2014SecondLook,McGregorVV/pods/2016/BetterTriangleCountingStreams,BeraC/stacs/2017/CliqueCountingEA,JowhariR26}. These works show that additional passes reduce the space complexity. 


The {\em Right to be Forgotten} streaming model was introduced in~\cite{PavanCVM/PODS/2024/RFDSIntro}, where streams may contain \emph{forget} operations that reset the frequency of an item to zero. Unlike standard turnstile deletions, forget operations are nonlinear and introduce significant algorithmic challenges. Subsequent work~\cite{LNSW25} developed near-optimal algorithms for estimating frequency moments in this model.  There has been no prior work on the Right-to-be-Forgotten in the context of graph streams. Triangle counting in dynamic graph streams, which allow edge insertions and deletions, has been studied in~\cite{KutzkovP14a, BulteauFKP16}. However, this model differs from the RFGS model: an edge that is already present cannot be inserted again, whereas in RFGS an existing edge may be inserted multiple times and a \emph{forget} operation removes all prior occurrences of the edge.
 
Prior works studied the space-complexity lowerbounds in the single edge arrival model~\cite{BarYossef2002Reductions, Braverman2013HowHard,BeraC/stacs/2017/CliqueCountingEA,KallaugherP17}. For one-pass algorithms, the best known lowerbound of $\Omega(m (\frac{\Delta_E}{T}+\frac{\sqrt{\Delta_V}}{T})$ is due to~\cite{Braverman2013HowHard,KallaugherP17}. The algorithm of~\cite{JK21} has a space complexity that matches this lowerbound.  Space-complexity lower bounds have also been studied for multipass algorithms also~\cite{BeraC/stacs/2017/CliqueCountingEA,Braverman2013HowHard,CormodeJ14,Jowhari2005NewStreaming}. Since the repeated edge arrival, and RFGS models generalize the single edge arrival model, the above-mentioned lowerbounds hold in these models too.

\ 

\noindent \textbf{Organization:}
We formally define the Repeated-edge Arrival graph streaming model and the Right to be Forgotten Graph Stream (RFGS) model, together with the necessary preliminaries, in Section~\ref{sec:prelim}. Our single-pass algorithms are presented in Section~\ref{Sec: Upper Bounds}: the algorithm for the Repeated-edge Arrival model is given in Section~\ref{SubSec: Single Pass Multi Occurrence}, while the algorithm for the RFGS model, along with a matching lower bound, is given in Section~\ref{SubSec: Single Pass Forget Model}. Multi-pass algorithms are presented in Section~\ref{sec:multi}, with the 2-pass algorithm in Section~\ref{sec:2PassMain} and the 4-pass algorithm in Section~\ref{sec:4Pass}. Extensions of the 2-pass and 4-pass algorithms to clique counting are given in Appendix~\ref{AppSec: Cliques}.


\section{Preliminaries}\label{sec:prelim}

We consider  simple, unweighted, undirected graphs $\graph = (\vertexset,\edgeset)$, where 
$\vertexset$ is the vertex set of size $\size{\vertexset} = \vertexcount$ and $\edgeset$ is the edge set of size $\size{\edgeset} = \edgecount$. We assume $\vertexset = [\vertexcount]$ and $\edgecount \ge \vertexcount$. For a vertex $v\in\vertexset$, $N_v$ denotes its set of neighbours and $d_v = |N_v|$ denotes its degree. For an edge $\edge = \{u,v\}$, the {\em edge degree} is defined as $d_\edge = \min\{d_u, d_v\}$. Let
$\triangleset = \{\{u,v,w\} : (u,v),(v,w),(w,u) \in \edgeset\}$ be the set of triangles in the graph, and let $\numtriangle = |\triangleset|$. For a vertex $\vertex$ (resp. edge $\edge$), let $\numtriangle_\vertex$ (resp. $\numtriangle_\edge$) be the number of triangles incident on $\vertex$ (resp. $\edge$). Let $\vertexsensitivity$ ($\edgesensitivity$) 
denote the vertex (resp., edge) sensitivity - the maximum number of triangles incident on any vertex (resp. edge). 





\subsection*{Streaming Algorithms}

\begin{definition}[Repeated-edge arrival graph streaming model]
Let $V$ be a vertex set. A graph stream over $V$ is a sequence of edges
$
\sigma=(e_1,\ldots,e_\ell)$, $e_i\in \binom{V}{2},
$
where an edge may appear multiple times. The stream represents the {\em simple graph}
$
G_\sigma=(V,E_\sigma),
E_\sigma=\{e : e=e_i \text{ for some } i\in[\ell]\}.
$. We use $m$ to denote the number of edges in $G_\sigma$.
\end{definition}

\begin{definition}[Right to be Forgotten Graph Stream (RFGS) model]
Let $V$ be a vertex set. A RFGS stream over $V$ is a sequence of updates
$\sigma=((e_1,U_1),\ldots,(e_\ell,U_\ell)),
e_i\in \binom{V}{2}, U_i\in\{+,\perp\}.$

An update $(e,+)$ inserts an occurrence of the edge $e$, while an update
$(e,\perp)$ \emph{forgets} $e$, deleting all occurrences of $e$ that
appeared prior to this update. 
Formally, the graph stream represents the {\em simple graph} $G_\sigma=(V,E_\sigma)$ where
\begin{align*}
E_\sigma&=\{e : \exists i \text{ such that } i \leq \ell, (e,+) = (e_i,+)\text{ and }\nexists j \text{ such that } i \leq j \leq \ell, (e,\perp) = (e_j,\perp)\}.
\end{align*}
Let $m$ denote the number of edges in $G_\sigma$, and $m_{\max{}}$ is defined as $\max_{i \in [t]} \size{E_{\sigma_i}}$ where $\sigma_i = \fbrac{(e_1,U_1),\ldots,(e_i,U_i)}$.
\end{definition}



In both repeated-edge and RFGS models, we are interested in estimating the number of triangles of the underlying simple graph. We assume that the algorithm is given bounds $\vsenseub, \esenseub$ and $\numtlb$ as inputs and the graph stream satisfies the following: i) $\vsenseub \ge \vertexsensitivity$: upper bound on vertex sensitivity; ii)$\esenseub \ge \edgesensitivity$: upper bound on edge sensitivity; iii) $\numtlb \le \numtriangle$: lower bound on the final number of triangles at the end of the stream.

\begin{definition}[$\appcon$-Approximation Algorithm]
Given parameters $\varepsilon,\delta\in(0,1)$, an algorithm is said to output a
\emph{$(\varepsilon,\delta)$-approximation} of $T$, if it returns $\esttriangle$ satisfying
$
\Pr\big[(1-\varepsilon)T \le \esttriangle \le (1+\varepsilon)T\big] \;\ge\; 1-\delta.
$
Equivalently, the algorithm with probability at least $1-\delta$ returns $\esttriangle$ that satisfies
$|\esttriangle-T| \le \varepsilon T$.
\end{definition}

\begin{remark}The space bounds of our algorithms are expressed in terms of $\vsenseub, \esenseub$ and $\numtlb$. We note that this is standard in the literature on triangle counting~\cite{JK21,BeraC/stacs/2017/CliqueCountingEA,McGregorVV/pods/2016/BetterTriangleCountingStreams}. 
\end{remark}






\noindent \textbf{Probabilistic Tools: } 
We use the standard Chebyshev's Inequality and multiplicative Chernoff bounds; see, e.g.,~\cite{DubhashiP/Book/2009/RandAlgInequalities,Mitzenmacher_Upfal_2005}. We also use $k$-wise independent hash functions. 

\begin{definition}[$k$-wise independent hash family]\label{def:hash}
Let $\mathcal{U}$ be a universe and let $\mathcal{H}$ be a family of hash functions
mapping $\mathcal{U}$ to $[s]$.
The family $\mathcal{H}$ is \emph{$k$-wise independent} if for any distinct
elements $u_1,\dots,u_k \in \mathcal{U}$ and a function $h$ chosen uniformly at
random from $\mathcal{H}$, the random variables
$h(u_1), \ldots, h(u_k)$ are independent and uniformly distributed over $[s]$.
\end{definition}
In particular, $k$-wise independence implies that for any distinct
$u_1,\dots,u_k \in \mathcal{U}$,
\[
\Pr\big[h(u_1) = h(u_2) = \cdots = h(u_k)\big] = s^{-(k-1)}.
\]
Moreover, $k$-wise independent hash families admit compact representations. Each hash function
$h \in \mathcal{H}$ can be stored using
$O\big(k \log(|\mathcal{U}| \cdot s)\big)$
bits, and for every $x \in \mathcal{U}$, the value $h(x)$ can be computed in $O\big(k \log(|\mathcal{U}| \cdot s)\big)$ time and space.

We also use the classical median-of-means estimator amplification method; see~\cite{chakrabarti2015data}.

\begin{lemma}[Median-of-Means]\label{lem:MoM}
Let $\widehat{\theta}$ be an unbiased estimator of a quantity $\theta$, computable using $\bigo{s}$ space, and suppose $\widehat{\theta}$ has bounded variance. Then an $(\varepsilon,\delta)$-approximation to $\theta$ can be obtained using space
$
O\left(
\frac{s\,\Var(\widehat{\theta})}{\E[\widehat{\theta}]^2}
\cdot
\frac{1}{\varepsilon^2}
\log\!\frac{1}{\delta}
\right)
.
$
\end{lemma}

\section{Single-Pass Algorithms}\label{Sec: Upper Bounds}


In this section, we present the single-pass algorithm for triangle counting in the repeated-edge arrival model and the right to be forgotten graph streaming model. We first discuss the challenges with the repeated edge arrival and our approach to overcome these challenges.  

\subsubsection*{Challenges with Repeated Edges.}

Most of the sampling-based approaches for triangle counting attempt to sample each triangle of the graph with the same probability and increment a counter $c$ when a triangle is sampled/detected. Since each triangle is sampled with the same probability, an appropriate scaling of $c$ leads to an unbiased estimator. The best-known triangle counting algorithm of Jayaram and Kallaugher~\cite{JK21} works as follows:  Initially, each vertex is sampled into the set $S_v$ with probability $p$, and when an edge $uv$ arrives, (1) $uv$ is sampled with into a set $S_E$ probability $q$, (2) a counter $c$ is incremented for every $w$ for which 
 $w \in S_V$, and both $uw$ and $vw$ are sampled into $S_E$. It is easy to see that each triangle causes the counter to increment with probability $pq^2$.


While the above algorithm has optimal space complexity in the standard single-arrival streaming model, the repeated-edge arrival model introduces significant challenges. In this setting, the probability that a triangle is detected can depend on the order as well as the multiplicity of edge occurrences in the stream. Consequently, just incrementing a counter whenever a new triangle is detected no longer yields a uniform sampling probability across triangles, resulting in a biased estimator.

We illustrate this by the following example. For simplicity, we assume the vertex sampling probability $p$ is $1$. 
In the repeated-edge arrival model, edges that appear multiple times can have a higher chance of being sampled, leading to a biased estimator. This issue is somewhat easy to address by using a hash function to sample the edges consistently throughout the stream. We show that a $4$-wise independent hash function suffices. However, the more subtle issue is that different triangles may cause the counter to increment with different probabilities. 
 Consider the following stream $ \langle ab, bc, bc, ac, uw, uv, vw, uv \rangle$. This stream has two triangles $abc$ and $uvw$. The triangle $abc$ is detected when both $ab$ and $bc$ are sampled (denoted by the event `$ab \wedge bc$'), which happens with probability $q^2$, and $ac$ completes it. However, the triangle $uvw$ is detected when $uw$ is sampled, and at least one of $uv$ or $vw$ is sampled: that is, the event `$uw \wedge (uv \vee vw)$' happens. The probability of this event is $2q^2 - q^3$. Thus, different triangles are detected with different probabilities. Since each time a triangle is detected, the counter $c$ is incremented, and $c$ will not be an unbiased estimator. In general, the expected value of $c$ depends on the stream order rather than the number of triangles in the underlying graph.

To overcome this, we employ a {\em dynamic triangle-sampling mechanism}.  We maintain a sample of triangles in a {\em non-monotone} manner: a triangle placed in the sample earlier may be removed from the sample later. We elaborate on this now. As mentioned, in addition to keeping samples of vertices ($S_v$) and edges ($S_E$), we keep a sample $S_T$ of detected triangles.  When an edge $e=uv$ comes, we {\em delete} all the triangles $\tau$ in $S_T$ for which  $e \in \tau$. Now, if $uv$ forms a triangle with two edges from $S_E$, and a vertex from $S_v$, then we place that triangle back in $S_T$ with probability $r$. Even though the idea of deleting triangles feels counter-intuitive, our analysis shows that each triangle is placed into $S_T$ with equal probability. 

We illustrate this on the above stream,$ \langle ab, bc, bc, ac, uw, uv, vw, uv \rangle$, with $p$ and $r$ equal to 1. As before, the triangle $abc$ is placed into $S_T$ with probability $q^2$. For the triangle $uvw$ to be placed into $S_T$, it must be the case that $uw$ must belong to $S_E$, which happens with probability $q$. Conditioned on $uw \in \samplededges$, consider the following disjoint events.

\begin{enumerate}
    \item `$uv \wedge \neg vw$': This happens with probability $q(1-q)$. When we see the edge $vw$, since it forms a triangle with sampled edges $uv$ and $vw$, we place the triangle $uvw$ in $S_T$. However, when we see the second occurrence of $uv$, this triangle is removed. So in this case, $uvw$ is placed in $S_T$ with probability zero.
    \item  `$\neg uv \wedge vw$':  In this case, the second occurrence of $uv$ will cause the triangle $uvw$ to go into $S_T$. Thus $uvw$ is placed in $S_T$ with probability $q(1-q)$.
    \item `$uv \wedge vw$':  In this case, when we first see $vw$, the triangle $uvw$ is placed in $S_T$. But when we see $uv$, we first remove $uvw$ from $S_T$. Since $uv$ forms a triangle with sampled edges $uw$ and $vw$, we place the triangle $uvw$ back in $S_T$. Thus, $uvw$ is placed into $S_T$ with probability $q^2$.
    \item `$\neg uv \wedge \neg vw$': both $uv$ and $vw$ are not sampled the triangle is in $S_T$ with probability zero.
\end{enumerate}
Thus at the end of the stream $uvw$ is in $S_T$ with probability $q[q(1-q) + q^2] = q^2$. Observe that when first or  last event happens the triangle remains in $S_T$ with probability $0$. The union of the second and third events is exactly the event $vw$ is sampled. Thus in general a triangle $uvw$, with $uv$ being the last edge occurrence in the stream, is sampled into $S_T$ (with probability $r$) if and only if all of  $uw$ $vw$, and $w$ are sampled which happens with probability $pq^2$.

\subsection{Repeated-Edge Arrival Model}\label{SubSec: Single Pass Multi Occurrence}

We now present our algorithm for the repeated-edge arrival graph streaming model (\Cref{Algo: Triangle with Multiple Occurrence with Bounds}). The algorithm is provided with the number of vertices $\vertexcount$, a lower bound on the number of triangles $\numtlb$, an upper bound on the vertex sensitivity $\vsenseub$, an upper bound on the edge sensitivity $\esenseub$, and outputs an estimate $\esttriangle$ of the number of triangles $\numtriangle$.

\begin{algorithm}[ht]
    \caption{Triangle Counting with Multiple Occurrence of Edges}\label{Algo: Triangle with Multiple Occurrence with Bounds}
    \begin{algorithmic}[1]
        \Require $\vertexcount, \numtlb, \vsenseub, \esenseub$ such that $\numtlb \leq \numtriangle$, $\vsenseub \geq \vertexsensitivity$, and $\esenseub \geq \edgesensitivity$
        \State $\vsamprate \gets \frac{\vsenseub}{\numtlb}$
        \State $\esamprate \gets \frac{\esenseub}{\vsenseub}$ if $\esenseub \geq \sqrt{\vsenseub}$, and $\frac{1}{\sqrt{\vsenseub}}$ otherwise
        \State $\tsamprate \gets \frac{\vsenseub}{\esenseub^2}$ if $\esenseub \geq \sqrt{\vsenseub}$, and $1$ otherwise
        \State Initiate $4$-wise independent hash function $\func{\edgehash}{[\vertexcount]\times[\vertexcount]}{\tbrac{\floor{{\frac{1}{\esamprate}}}}}$ 
        \State Initiate $2$-wise independent hash function $\func{\vertexhash}{[\vertexcount]}{\tbrac{\floor{{\frac{1}{\vsamprate}}}}}$ 
        \State $\samplededges \gets \emptyset$; $\sampledtriangles \gets \emptyset$
        \For{each update $(\altvertex,\vertex)$}\label{SPRE Line: Update Start}
            \For{$\triangle \in \sampledtriangles$}\label{Loop:Cleaning}
                \If{$(\altvertex,\vertex) \in \triangle$}
                    \State $\sampledtriangles \gets \sampledtriangles\setminus\triangle$
                \EndIf
            \EndFor
            \For{$w \in \vertexset$}\label{Loop:Triangle Samp}
                \If{$(w,\altvertex),(w,\vertex) \in \samplededges$ \textbf{and} $\vertexhash(w) = 1$}
                    \State $\sampledtriangles \gets \sampledtriangles \cup \sbrac{w,\altvertex,\vertex} $ w.p. $\tsamprate$
                \EndIf
            \EndFor
            \If{($\vertexhash(\vertex) = 1$ \textbf{or} $\vertexhash(\altvertex) = 1$) \textbf{and} $\edgehash((\altvertex,\vertex)) = 1$}\label{Line:Edge Samp}
                \State $\samplededges \gets \samplededges \cup (\altvertex,\vertex)$
            \EndIf
        \EndFor\label{SPRE Line: Update End}
        \State $\esttriangle \gets \size{\sampledtriangles}/\vsamprate\esamprate^2\tsamprate$
        \Return $\esttriangle$
    \end{algorithmic}
\end{algorithm}

First, we show that our algorithm ensures each triangle is sampled only through the last edge of the triangle that arrives in the stream.

\begin{lemma}\label{Lemma: Order Enforced}
    In~\Cref{Algo: Triangle with Multiple Occurrence with Bounds}, for any triangle $\sbrac{w,\altvertex,\vertex}$, if $(\altvertex,\vertex)$ is the last of its edges to appear in the stream, then it is sampled to $\sampledtriangles$ only through the edge $(\altvertex,\vertex)$.
\end{lemma}

\begin{proof}
    Let us consider the last arrival of the edge $\fbrac{\altvertex,\vertex}$ in the stream. Line~\ref{Loop:Cleaning} ensures that if the triangle $\sbrac{w,\altvertex,\vertex}$ was added to the $\sampledtriangles$ earlier, the occurrence is removed from the sample. Then, it can be only sampled through Line~\ref{Loop:Triangle Samp}. As this was the last occurrence of any edge in the triangle, the algorithm does not add it to $\sampledtriangles$ for any of the consequent edges in the stream.
\end{proof}

Next, we show that $\esttriangle$ is an unbiased estimator of the number of triangles, $\numtriangle$.

\begin{lemma}\label{Lemma: MultiOccur Unbiased Estimator}
    $\E[\esttriangle] = \numtriangle$
\end{lemma}

\begin{proof}
    \Cref{Lemma: Order Enforced} ensures that any triangle is sampled exactly once, through the last edge of the triangle in the stream. Let us denote $\sequence{X}{\numtriangle}$ to be the indicator random variables denoting whether the $i$-th triangle is in $\sampledtriangles$ at the end of the stream. Then, we have $X = \sum_{i \in \numtriangle} X_i = \size{\sampledtriangles}$. For a triangle to be detected through its last edge, the opposite vertex has to be mapped to $1$ by $\vertexhash$, the other two edges have to be mapped to $1$ by $\edgehash$. Additionally, it is sampled to $\sampledtriangles$ with probability $\tsamprate$. Then, we have:
   $ \E[\size{\sampledtriangles}] = \E[X] = \E\tbrac{\sum_{i \in \numtriangle} X_i} = \sum_{i \in \numtriangle} \E[X_i] = \numtriangle \Pr[X_i = 1] = \numtriangle\vsamprate\esamprate^2\tsamprate$.
    Hence, we have:
    $\E[\esttriangle] = \E\tbrac{\frac{\size{\sampledtriangles}}{\vsamprate\esamprate^2\tsamprate}} = \frac{1}{\vsamprate\esamprate^2\tsamprate}\E\tbrac{\size{\sampledtriangles}} = \numtriangle$. \qedhere
\end{proof}

The next lemma bounds the variance of the estimator $\esttriangle$. 

\begin{lemma}\label{Lemma: MultiOccur Estimator Variance}
    $\Var[\esttriangle] \leq\frac{\numtriangle}{\vsamprate\esamprate^2\tsamprate}+\frac{\numtriangle\edgesensitivity}{\vsamprate\esamprate}+\frac{\numtriangle\vertexsensitivity}{\vsamprate}$
\end{lemma}

From Lemma~\ref{Lemma: MultiOccur Estimator Variance}, we have the following corollary by bounding the variance under our choices of $\vsamprate$, $\esamprate$, and $\tsamprate$. The proof is presented in Appendix~\ref{Sec:1passextra}.

\begin{corollary}\label{Cor: MultiOccur Variance Control}
    For the values of $\vsamprate$, $\esamprate$, and $\tsamprate$ of~\Cref{Algo: Triangle with Multiple Occurrence with Bounds}, we have $\Var[\esttriangle] \leq 4\numtriangle^2$.
\end{corollary}

Now, we analyse the space complexity of \Cref{Algo: Triangle with Multiple Occurrence with Bounds}.

\begin{lemma}\label{Lemma: Multi Occur Space Complexity}
    The space complexity of~\Cref{Algo: Triangle with Multiple Occurrence with Bounds} is $\bigot{\edgecount\fbrac{\frac{\esenseub}{\numtlb}+\frac{\sqrt{\vsenseub}}{\numtlb}}}$.
\end{lemma}

\begin{proof}
    The algorithm stores the hash functions $\vertexhash$ and $\edgehash$, sample of edges $\samplededges$, and sample of triangles $\sampledtriangles$. The hash functions $\vertexhash$, and $\edgehash$ can be stored using $O(\log n)$ space. Next, we consider the edge samples in $\samplededges$. The algorithm stores $2\edgecount\vsamprate\esamprate$ edges in expectation. This takes $\bigot{\edgecount\fbrac{\frac{\esenseub}{\numtlb}+\frac{\sqrt{\vsenseub}}{\numtlb}}}$ space.

    If $\esenseub < \sqrt{\vsenseub}$, then there are $\numtriangle\vsamprate\esamprate^2$ triangles in $\sampledtriangles$ in expectation. Given the values of $\vsamprate$ and $\esamprate$, this takes $\bigot{\frac{\numtriangle}{\numtlb}} = \bigot{\frac{\edgecount\esenseub}{\numtlb}}$ space. If $\esenseub \geq \sqrt{\vsenseub}$, the algorithm stores $\numtriangle\vsamprate\esamprate^2\tsamprate$ triangles in expectation. This takes $\bigot{\frac{\numtriangle}{\numtlb}} = \bigot{\frac{\edgecount\esenseub}{\numtlb}}$ space in expectation. 
\end{proof}

    Combining the results above, and~\Cref{lem:MoM} gives us the following result.

\begin{theorem}\label{Theorem: MultiOccur AppCon Estimator}
    There is a single-pass streaming algorithm in the repeated edge occurrence model, that, given $\vertexcount$, $\edgecount$, $\vsenseub$, $\esenseub$, and $\numtlb$, outputs an $\appcon$-estimate $\esttriangle$ of the number of triangles $\numtriangle$ of the graph using $\bigot{\edgecount\fbrac{\frac{\esenseub}{\numtlb}+\frac{\sqrt{\vsenseub}}{\numtlb}}\frac{\log(1/\confidence)}{\approxerror^2}}$ space.
\end{theorem}

\begin{proof}
    We use the Median-of-Means estimator by combining the individual estimators as outlined above. \Cref{Lemma: MultiOccur Unbiased Estimator} ensures that the estimator is unbiased, and is $\numtriangle$ in expectation. By~\Cref{Cor: MultiOccur Variance Control}, we have the variance bounded above by $4\numtriangle^2$. And by~\Cref{Lemma: Multi Occur Space Complexity}, individual estimators use $\bigot{\edgecount\fbrac{\frac{\esenseub}{\numtlb}+\frac{\sqrt{\vsenseub}}{\numtlb}}}$ space. Combining these facts using~\Cref{lem:MoM} completes the proof.
\end{proof}

\subsection{RFGS Model}\label{SubSec: Single Pass Forget Model}

In this section, we outline our algorithm for the right to be forgotten model.  The basic structure of the algorithm remains similar to that of \Cref{Algo: Triangle with Multiple Occurrence with Bounds} with some appropriate changes needed for the update steps (lines~\ref{SPRE Line: Update Start} to \ref{SPRE Line: Update End}, in \Cref{Algo: Triangle with Multiple Occurrence with Bounds}).

In the right to be forgotten model, the updates are of form $\langle(\altvertex,\vertex), \Delta\rangle$ with $\Delta \in \{+, \perp\}$.  
When an update with $\Delta = +$ comes then the updates are done exactly how it is done in Lines~\ref{SPRE Line: Update Start} to \ref{SPRE Line: Update End} of \Cref{Algo: Triangle with Multiple Occurrence with Bounds}. When the update has $\Delta = \perp$ then we need to do the following: 


\begin{algorithm*}[h]
\hrulefill
    \begin{algorithmic}
        \For{each update $\langle(\altvertex,\vertex),\perp\rangle$}
            \For{$\triangle \in \sampledtriangles$}
                    \If{$(\altvertex,\vertex) \in \triangle$}
                        \State $\sampledtriangles \gets \sampledtriangles\setminus\triangle$
                    \EndIf
            \EndFor
            \If{$(\altvertex,\vertex)$}
                \State $\samplededges \gets \samplededges\setminus{(\altvertex,\vertex)}$
            \EndIf
        \EndFor
    \end{algorithmic}
\end{algorithm*}


The complete pseudocode of the \textit{Triangle Counting - RFGS Model Algorithm} is presented (as  \Cref{Algo: Triangle with Forget Model with Bounds}) in the Appendix~\ref{Appendix: ForgetCode}.
The proof of correctness is also similar to that of \Cref{Algo: Triangle with Multiple Occurrence with Bounds} with some slight appropriate modifications.  Below we present the proof of correctness of the \Cref{Algo: Triangle with Forget Model with Bounds}.



Note that the number of triangles in a graph that is realized at any intermediate point in the stream can be much larger than the number of triangle in the final graph. 
Let $\numtriangle_{\max}$ denote the maximum number of triangles in any graph that is realized at any point during the stream. The analysis of the updated algorithm follows similarly to the analysis of~\Cref{Algo: Triangle with Multiple Occurrence with Bounds}. The key changes are in~\Cref{Lemma: Order Enforced} and~\Cref{Lemma: Multi Occur Space Complexity}, as outlined below.

\begin{lemma}\label{Lemma: Forget Enforced} 

In \Cref{Algo: Triangle with Forget Model with Bounds},
    only the triangles $\sbrac{w,\altvertex,\vertex}$ that survives at the end of the stream will be sampled in $\sampledtriangles$. Furthermore, for any triangle $\sbrac{w,\altvertex,\vertex}$ that survives at the end of the stream, if $(\altvertex,\vertex)$ is the last of its edges to appear in the stream, then it is sampled to $\sampledtriangles$ only through the edge $(\altvertex,\vertex)$. 
\end{lemma}

\begin{proof}[Proof of Lemma~\ref{Lemma: Forget Enforced}]
    If a triangle included in $\sampledtriangles$ has any of its edges forgotten in a latter point of the stream, the modification ensures that this triangle will be excluded from $\sampledtriangles$. Any triangle that survives at the end of the stream is detected through the last of its edges in the stream, by the same arguments as~\Cref{Lemma: Order Enforced}. 
\end{proof}

\begin{lemma}\label{Lemma: Forget Space Complexity}
     \Cref{Algo: Triangle with Forget Model with Bounds} has space complexity $\bigot{\edgecount_{\max{}}\fbrac{\frac{\esenseub}{\numtlb}+\frac{\sqrt{\vsenseub}}{\numtlb}}+\frac{\numtriangle_{\max{}}}{\numtriangle}}$.
\end{lemma}

\begin{proof}[Proof of Lemma~\ref{Lemma: Forget Space Complexity}]
    The space complexity for the hash functions and edges follows similarly as in~\Cref{Lemma: Multi Occur Space Complexity}. Hence, we focus on the space required to store the sampled triangles here.

    If $\esenseub < \sqrt{\vsenseub}$, then there are $\numtriangle\vsamprate\esamprate^2$ triangles in $\sampledtriangles$ in expectation. Given the values of $\vsamprate$ and $\esamprate$, this takes $\bigot{\frac{\numtriangle_{\max}}{\numtlb}}$ space. If $\esenseub \geq \sqrt{\vsenseub}$, the algorithm stores $\numtriangle\vsamprate\esamprate^2\tsamprate$ triangles in expectation. For the corresponding values of $\vsamprate, \esamprate$ and $\tsamprate$, this takes $\bigot{\frac{\numtriangle_{\max}}{\numtlb}}$ space in expectation. 
\end{proof}

The rest of the analysis follows similarly as the repeated edge occurrence model .

\begin{theorem}\label{Theorem: Forget AppCon Estimator}
    There is a single-pass streaming algorithm for the RFGS model, that, given $\vertexcount$, $\edgecount$, $\vsenseub$, $\esenseub$, and $\numtlb$, outputs an $\appcon$-estimate $\esttriangle$ of the number of triangles $\numtriangle$ of the graph using $\widetilde{O}\Big(\Big(\edgecount_{max}\Big(\frac{\esenseub}{\numtlb}+\frac{\sqrt{\vsenseub}}{\numtlb}\Big)+\frac{\numtriangle_{\max{}}}{\numtriangle}\Big)\frac{\log(1/\confidence)}{\approxerror^2}\Big)$ space.
\end{theorem}
\begin{remark}
Another natural forget operation is forgetting a vertex $v$, which corresponds
to removing all edges incident on $v$ from the graph. This can be simulated
within our model by issuing a forget update $\langle (v, u), \perp \rangle$
for each $u \in V$. This is because our algorithm works even when a
forget update arrives for an edge that has not yet appeared in the stream.

\end{remark}

\subsubsection{Lower Bounds for the RFGS model}

The following theorem establishes the optimality of \Cref{Algo: Triangle with Forget Model with Bounds}. 

\begin{theorem}\label{Thm: RFGS Lower bound}
    Any algorithm that obtains a $(\nicefrac{1}{3},\nicefrac{1}{3})$-estimate of the number of triangles in the RFGS model must use $\bigomega{\frac{\numtriangle_{\max{}}}{\numtriangle}+\frac{m_{\max{}}\edgesensitivity}{\numtriangle}+\frac{m_{\max{}}\sqrt{\vertexsensitivity}}{\numtriangle}}$ space.
\end{theorem}

To prove the lower bound we appeal to the result of~\cite{PavanCVM/PODS/2024/RFDSIntro} who showed that any streaming algorithm in the RFDS model that estimates $F_0$ requires $\Omega\left (\frac{F_0(\bS')}{F_0(\bS)}\right)$ space, where $\bS'$ is the stream obtained by removing forget operations from $\bS$. We reduce the $F_0$ estimation to triangle estimation by generating a graph stream $\bS_\graph$ from stream $\bS$. For each element $(u, +) \in \bS$, we create a distinct clique of size $K$ and add the edges of that clique to $\bS_\graph$. If $(u, \bot)$ appears in the stream $\bS$, we forget all the edges corresponding to the clique. 
The stream $\bS_\graph$ will have the following properties: $\edgesensitivity$ equals $K-1$, $T(\bS_\graph)$ equals 
$F_0(\bS)\binom{K}{3}$ and $m_{\max}$ is atmost  $F_0(\bS')\binom{K}{2}$. Now a $o\left(\frac{m_{\max}\edgesensitivity}{T}\right)$-space algorithm for triangle estimation implies a $\Omega\left (\frac{F_0(\bS')}{F_0(\bS)}\right)$ to estimate $F_0$, which is a contradiction. The details are given in \Cref{AppSec: LB}.



\section{Multi-Pass Algorithms}\label{sec:multi}
In this section, we present two multi-pass algorithms for triangle-counting in graph streams that allow repeated occurrences of edges. Recall that $\numtriangle$ is the number of triangles in an $\vertexcount$-vertex $\edgecount$-edge input graph $G = (V,E)$ and $\numtlb$ is a \emph{promised} lower bound on $T$. Our two-pass algorithm is suitable for graphs  with  the  $T \le m$, \emph{scarce-regime}, henceforth and our four-pass algorithm is suitable for triangles with $T > m$, \emph{abundant-regime}, henceforth.

The reason for having two different algorithms depending on the number of triangles in the input graph is as follows. In the scarce-regime any uniform edge or wedge (length-two path) sampling strategy incurs a large variance in the estimated triangle count caused by certain \emph{heavy} edges (that participate in many triangles). This issue is elaborated in more detail in Appendix~\ref{sec:twoPass}.
Any algorithm for the scarce-regime can control the variance by estimating the contribution of the heavy edges to the overall triangle count separately from the \emph{light} (that participate in a few triangles) edges. This approach was previously used by McGregor et al. in ~\cite{McGregorVV/pods/2016/BetterTriangleCountingStreams}. We extend their two-pass algorithm for arbitrary order streams (see Section $3.1$ and Algorithm TRIANGLES3 from~\cite{McGregorVV/pods/2016/BetterTriangleCountingStreams}) to graph streams with repeated edge occurrences.

In the abundant regime, following~\cite{BeraC/stacs/2017/CliqueCountingEA}, we bound the variance by imposing a total order on the vertices based on their estimated degrees, using vertex IDs to break ties. This ordering lets us assign each triangle $\tau = (u,v,w)$ in $G$ to a unique vertex of $\tau$ (either its minimum-degree or maximum-degree vertex), ensuring that every triangle can only be sampled and counted in a single canonical way. This limits the influence of high-degree vertices on the overall variance. In addition to the total order on the vertices, we use the standard technique 
of variance reduction by repeated sampling.
 The details of how  we build upon and extend the four-pass algorithm of~\cite{BeraC/stacs/2017/CliqueCountingEA}
 are presented in Section~\ref{sec:4Pass}.

 In both the regimes, our algorithms involve maintaining certain counters, for example, estimating the degrees of certain vertices as the stream progresses. Using a simple counter 
 and updating it every time, we witness an increment based on the updates seen on the stream no longer works, because an edge $e= \{u, v\}$ that triggers the increment may appear multiple times and incrementing the counter every time it
 appears can lead to over-counting. To take care of such issues, we instead use $F_0$ counters based on~\cite{cohen1997size, cohen2008tighter} which count the number of distinct elements.

We use a slight modification of the bottom-$k$ distinct-elements estimator of~\cite{cohen1997size, cohen2008tighter}, also known as the $\KMV$ estimator (where $\KMV$ stands for $k$-minimum values). For our application, we require the additional guarantee that the output of the estimator is bounded with probability $1$, a property that is not explicitly guaranteed by the estimators in~\cite{cohen1997size, cohen2008tighter}. We therefore introduce a clipped version of the $\KMV$ estimator, which we call the $\cKMV$ \emph{estimator}.

Let $\hat F_0$ denote the output of the $\KMV$ estimator of~\cite{cohen1997size, cohen2008tighter}, as described in Theorem~\ref{thm:F_0BJKST}, for estimating the number of distinct elements in a stream over a universe $\mathcal U$ of size $U$. The $\cKMV$ estimator outputs
$\min\{\hat F_0,2U\}$.
Thus, in the event that the original $\KMV$ estimator outputs a value larger than $2U$, we simply truncate its output to $2U$. This guarantees deterministically that the output is at most $2U$. In our application, $U=n-1$, corresponding to the maximum possible degree of a vertex. In Appendix~\ref{sec:ClippedEstimator}, we explicitly analyze the expectation, variance, and tail bounds of the $\cKMV$ estimator. The resulting guarantees are stated in Theorem~\ref{thm:F_0BJKST Clipped}.

\begin{theorem}[$\cKMV$ distinct-count estimator]\label{thm:F_0BJKST Clipped}
Let $ \mathcal U$ be a universe such that $|\mathcal U| = U$. 
Let a stream over $\mathcal U$ contain $F_0$ distinct elements, where $0 \le F_0 \le U$.
For any $k\geq 3$ there exists an estimator $\cKMV_{k}$ that on the stream $\mathcal{U}$ outputs $F_{clipped}$ with the guarantee
\begin{enumerate}
    \item \textbf{Expectation}
            $
                  F_0\left( 1- 2e^{-ck/2}) \right) \le \E[F_{clipped}] \le F_0.
            $
    \item \textbf{Variance}
    $
         \Var(F_{clipped}) \le \frac{F_0(F_0 - k + 1)}{k-2}.
    $
    \item \textbf{Tail Bound}
    $
         \Pr \left(\lvert F_{clipped} - F_0 \rvert  \ge \varepsilon F_0 \right) \le 2e^{-c k \varepsilon^2}.
    $
    \item \textbf{Space Complexity} of $F_{clipped}$ is $O(k\log U)$
\end{enumerate}
\end{theorem}

\subsection{A Two-Pass Algorithm}\label{sec:2PassMain}
We give the pseudocode (Algorithm~\ref{Algo: Triangle with 2-pass multi-occurrence}) of the two-pass algorithm and state the space bound here.
The detailed description and analysis of  the two-pass algorithm is given in the Appendix~\ref{sec:twoPass}.

\begin{algorithm}[ht]
    \caption{Two-pass Triangle Counting for Repeated Edge Arrival}\label{Algo: Triangle with 2-pass multi-occurrence}
    \begin{algorithmic}[1]
        \Require $\vertexcount, \numtlb$
        \State $\vsamprate \gets \frac{\beta \log n}{\varepsilon^2 \sqrt{\numtlb}}$
         \State $V_{oracle} \gets \emptyset$, $E_{oracle} \gets \emptyset$, $E_{sampled} \gets \emptyset$
         \State Initiate $V_{oracle}$ with a random subset of $V$ such that every vertex  $v \in V$ is added to $V_{oracle}$ with probability $p$.
         \State Initiate $3$-wise independent hash function $\func{\edgehash}{[\vertexcount]\times[\vertexcount]}{\tbrac{\lfloor {\frac{1}{\vsamprate}\rfloor }}}$ 
         \State \textbf{Pass 1: }
          \Comment{Build the oracle and populate the edge set $E_{sampled}$}
          \For{each update $(u,v)$}
               \If{$(u \in V_{oracle} \lor v \in V_{oracle}) \land (u,v) \notin E_{oracle})$}
                        \State $E_{oracle} \gets E_{oracle} \cup \{(u,v)\}$
               \Comment{Sample the edge into $E_{Sampled}$}
                \If{$h((u,v)) = 1$ and $(u,v) \notin E_{sampled}$} 
                         \State $E_{sampled} \gets E_{sampled} \cup \{(u,v)\} $
                \EndIf
                \EndIf
           \EndFor
          \State \textbf{Pass 2: }
          \State We run different $\cKMV$ estimators on streams $\mathcal{F}_L$, $\mathcal{F}^{(1)}_H$, $\mathcal{F}^{(2)}_H$ and  $\mathcal{F}^{(3)}_H$
          \For{each update $(u,v)$}
                \If{$\textsf{oracle}((u,v)) = L$}\label{lin:lightlabelled}
                    \State  
                    For each $\tau$ in $\{(u,v,w) \mid \{u,w\}, \{v,w\} \in E^L_{sampled} \}$ add $\tau$ to the stream $\mathcal{F}_L$. \label{lin:LightDetected}                          
                \EndIf
                 \If{$\textsf{oracle}((u,v)) = H$}
                    \State For each $ \tau$ 
                    in  $\{ (u,v,w) \mid \{u,w\}, \{v,w\} \in E^L_{oracle} \}$ add $((u,v), \tau)$ to the stream $\mathcal{F}^{(1)}_H$  
                    \State  For each $\tau$ in
                    $\{(u,v,w) \mid \{u,w\},  \{v,w\} \in E^L_{oracle} \cap E^H_{oracle}  \}$ add $((u,v), \tau)$ to the stream $\mathcal{F}^{(2)}_H$.
                    \State For each $\tau$ 
                    in  $\{(u,v,w) \mid \{u,w\}, \{v,w\} \in E^H_{oracle} \}$ add $((u,v), \tau)$ to  stream $\mathcal{F}^{(3)}_H$.                      
                \EndIf                      
          \EndFor
           \State $F_L$, $F^{(1)}_H$, $F^{(2)}_H$ and $F^{(3)}_H$ be the final output of the $\cKMV$ estimators on streams $\mathcal{F}_L$, $\mathcal{F}^{(1)}_H$, $\mathcal{F}^{(2)}_H$ and $\mathcal{F}^{(3)}_H$ respectively
           \State $\esttriangle = F_L/(3p^2 - 2p^3) + (\sum_{i} F^{(i)}_H/i )/p$. \\       
        \Return $\esttriangle$
    \end{algorithmic}
\end{algorithm}

\begin{restatable}[]{theorem}{TwopassTriangle}\label{Thm:2Pass}
    There exists an $\tilde O(\varepsilon^{-2}m/\sqrt T )$-space $2$-pass algorithm for the repeated edge model  that returns a $(1 \pm \varepsilon)$-approximation of $T$ with probability at least $96/100$.
\end{restatable}

\subsection{A Four-Pass Algorithm}\label{sec:4Pass}
Our four-pass algorithm is adapted from the four-pass algorithm of~\cite{BeraC/stacs/2017/CliqueCountingEA}. We will first describe the triangle counting algorithm of~\cite{BeraC/stacs/2017/CliqueCountingEA} and then show how this can be adapted to work with graphs streams with multiple occurrences of edges. We will refer to the four-pass triangle counting algorithm of~\cite{BeraC/stacs/2017/CliqueCountingEA} as 
BC henceforth.
BC algorithm considers a total order $v_1 \prec v_2 \prec \ldots v_n$ of the vertex set $V$, where $u \prec v$ if $d_u < d_v$, or $d_u = d_v$ and ID of $u$ is less than the ID of $v$.
Given the above total order, consider a triangle $\tau = (u, v,w)$ such that $u \prec v \prec w$, then $\tau$ can be uniquely associated to the edge $\{u,v\}$ incident on the first two vertices in the above order.
For a given edge $e$, let $y_e$ denote the number of triangles in $G$ associated with edge $e$. It follows that that $\sum_{e\in E} y_e = T$.

The algorithm of~\cite{BeraC/stacs/2017/CliqueCountingEA}
proceeds as follows:
\begin{enumerate}
\item In the \textbf{first pass} sample a single edge $e \in E$ uniformly at random from the stream.
\item In the \textbf{second pass} compute the degrees of the two endpoints of the sampled edge $e = \{u,v\}$.
\item Sample a random neighbour $w$ of $u$
 (minimum degree endpoint of previously sampled edge) in the 
 \textbf{third pass}. Observe that this gives us a wedge or
 $2$-length path.
 \item In the \textbf{fourth pass} estimate the degree of $w$ and check if $(u,v,w)$ is a triangle and whether $u \prec v \prec w$
and if so a random variable $Z$ is assigned the minimum of
the degrees of the two endpoints of the edge $e$, i.e., 
$Z = \min \{d_u, d_v\}$.
\end{enumerate}
One can show that the expected value of the random variable $Z$ is $y_e$. Then, by scaling this random variable by $m$, one gets a random variable with expectation $T$.
The problem here is the variance. BC controls the variance of the estimator by repeating the neighbour sampling process an appropriate number of times and then taking an average over the random variables assigned to these samples.

The main challenge in implementing this algorithm in 
graph streams with multiple occurrences is in computing the vertex degrees space efficiently. We can not simply use a counter to keep track of the number of edges incident to a vertex $v$ as the stream progresses since that can lead to over-counting due to multiple occurrences of certain edges. On the other hand the space required to explicitly store every new edge incident on a certain vertex can be prohibitively large for high degree vertices.
We rely on an $F_0$ counter to approximately count the number of distinct edges incident on $v$. One more change we make to handle the multiple occurrences is the use of $\ell_0$-samplers to uniformly sample an edge from the stream and for sampling the  neighbours of the sampled edge $e$. A major part of our analysis involves proving that using the approximate degrees instead of the exact degrees and the use of $\ell_0$-samplers only slightly changes the expected value of $\esttriangle$ and its variance is only slightly worse than that of the four-pass algorithm of~\cite{BeraC/stacs/2017/CliqueCountingEA}. Finally we show that the space usage of our algorithm is
 asymptotically same as that of~\cite{BeraC/stacs/2017/CliqueCountingEA}.

\noindent 

\noindent We assume access to an $\ell_0$-sampler that on a stream (with multiple occurrences of a given element)
defining a vector  $x \in \mathbb R^n$ returns an index 
$i \in [n]$ chosen uniformly at random from the support of $x$ with high success probability.
Specifically, we will be using the $\ell_0$ sampler of~\cite{JowhariSaglamTardos2011} as a blackbox.
The following theorem summarizes the performance guarantees of the $\ell
_0$ sampler of~\cite{JowhariSaglamTardos2011}.
\begin{theorem}[Theorem $2$ of~\cite{JowhariSaglamTardos2011}]\label{thm:Jowhari2011}
There exists a randomized streaming algorithm which, given an
underlying vector $x \in \mathbb R^n$ presented as a stream of 
coordinate updates, outputs an index $i \in [n]$ such that
$\Pr[i = j] = \frac{1}{\| x \|_0}$, if $x_j \neq 0$ and 
$0$, otherwise, with failure probability at most $\delta$.
The algorithm uses $O(\log^2 n \log (1/\delta))$ bits of space.
\end{theorem}

\begin{algorithm}[ht]
    \caption{Four-pass Triangle Counting for Repeated-Edge Arrival}\label{Algo: Triangle with 4-pass multi-occurrence}
    \begin{algorithmic}[1]
        \Require $\vertexcount, \edgecount$
         \State \textbf{Pass 1:}
          \State Select one edge $e = \{u,v\}$ uniformly at random from the stream
                   using an $\ell_0$ sampler.\label{Lin:EdgeSampled}
          \State \textbf{Pass 2:}
            \State Compute $\hat d_u$ and $\hat d_v$ using $\cKMV_{k}$ estimator with $k = c\varepsilon^{-2}(\log n + \log \frac{1}{\varepsilon})$ for $d_u$ and $d_v$. \label{lin:pivotDegreesEst}
            \State \textbf{Pass 3:}
              \State Let $r \gets \lceil \min\{\hat d_u, \hat d_v\} /(1-\varepsilon)\sqrt m \rceil$.\label{lin:mindegree}
              \State $pivot \gets \arg \min \{\hat d_u, \hat d_v \} $.
              \For { $k \gets 1 \text{~to~} r$}
                 \Comment{We use an $\ell_0$ sampler to sample uniformly from the set of distinct neighbours of $pivot$.}
                \State Sample a vertex $w_k$ using an $\ell_0$ sampler from $N_{pivot}$.
            \EndFor
             \State \textbf{Pass 4:}
               \State Compute $\hat d_{w_1},\hat d_{w_2}, \ldots, \hat d_{w_r}$ using $\cKMV_k$ estimators
               \For {$ \ell \gets 1 \text{~to~} r$}
                  \State  $Z_\ell \gets 0$.
                    \If{$(u,v,w_\ell) \text{~form a triangle and~} w_\ell \succ \max\{u, v \} $}\label{lin:traingleDetected4Pass}
                   \State  $Z_\ell \gets \hat  d_{pivot}$
                   \EndIf
               \EndFor
                \State $Y = \frac{1}{r} \sum_{\ell =1}^{r} Z_ell$
                \State $\esttriangle = m \cdot Y$
                \State \Return $\esttriangle$
    \end{algorithmic}
\end{algorithm}

Now, we analyse the \Cref{Algo: Triangle with 4-pass multi-occurrence}. 
First, we show that \Cref{thm:F_0BJKST Clipped} can be used to maintain good estimates of the degrees of all the vertex in the graph. 

\begin{restatable}[]{corollary}{ClipKMVUnion}\label{Cor: Clipped KMV Union}
    Let $\hat d$ denote the vector of degree estimates of all the vertices produced by the $\cKMV$-estimator of Theorem~\ref{thm:F_0BJKST Clipped}.
Let $\mathcal E_{tight}$ denote the event that, for every vertex $v \in V$, the estimates $\hat d_v$ of $\cKMV_k$ satisfies 
$(1 - \varepsilon) d_v \le \hat d_v \le (1 + \varepsilon) d_v$.
For $k = c\varepsilon^{-2}(\log n + \log \frac{1}{\varepsilon})$, we have
$\Pr(\mathcal{E}_{tight}) \geq 1 - \varepsilon n^{-6}$.
\end{restatable}

\begin{proof}[Proof of Corollary~\ref{Cor: Clipped KMV Union}]
By the tail-bound for the clipped KMV-estimator~(\Cref{Cor: Clipped KMV Union}),
and by a union bound over all $n$ vertices
$\Pr(\mathcal{E}_{tight}^c) \le 2 n e^{-c k \varepsilon^2}$.
Thus for appropriate choice of $k$,  by Theorem~\ref{thm:F_0BJKST Clipped}, the event $\mathcal E_{tight}$ holds with probability at least $1- \delta$.
For $k = \Theta\left(\varepsilon^{-2}(\log n + \log (1/\varepsilon)\right)$, we have
$\Pr(\mathcal{E}_{tight}^c) \le \varepsilon n^{-6}$.\qedhere
\end{proof}


The vector $\hat d$  induces a total order on the vertices (ties broken by IDs). With respect to this order, every triangle is uniquely associated with the edge joining its first two vertices. Throughout, we  let $y_e(\hat d)$ denote the number of triangles {associated with edge $e$ according to the order induced by degree estimates $\hat d$}. By construction, we have
$ \sum_{e \in E} y_e(\hat d) = T$. First, we analyse the expected value of our estimator.

\begin{lemma}\label{Lem:4Pass-Expected}
The estimator $\esttriangle$ produced by Algorithm~\ref{Algo: Triangle with 4-pass multi-occurrence} satisfies
$\E[\esttriangle] = T\fbrac{1 \pm\frac{\varepsilon}{2}}$.
\end{lemma}

\begin{proof}[Proof of Lemma~\ref{Lem:4Pass-Expected}]
Let $\mathcal E_{e}$ be the event that the edge
$e$ is sampled in Line~\ref{Lin:EdgeSampled} of \Cref{Algo: Triangle with 4-pass multi-occurrence}.
Every time a triangle is detected in line~\ref{lin:traingleDetected4Pass}, the variable $Z_\ell$ is assigned the estimated degree $\hat d_{pivot}$ of the pivot vertex. Let $A_{\ell}$ be the indicator 
that a clique with vertex $w_\ell$ is detected in line~\ref{lin:CliqueDetected4Pass}.
Then $\E[A_\ell \mid \mathcal E_{e}  ] = \frac{y_{e}}{d_{pivot}},$
where $y_{e}$ is the number of triangles associated with the edge $e$ according to the ordering defined by estimated degrees and $d_{pivot}$ is the true degree of the pivot vertex. We have the true degree in the denominator since the $\ell_0$ sampler will 
sample a neighbour uniformly from the neighbourhood of $pivot$ although we don't know the true value of its size. Then, we have
$\E[Z_{\ell} \mid \mathcal E_{e}]  = \hat d_{pivot} \cdot \frac{y_{e}}{d_{pivot}}$.

We now focus on the expectation of $\hat{d}_{pivot}$. By \Cref{Cor: Clipped KMV Union}, we have
\begin{align}
    \E[\hat{d}_{pivot}] &= \E[\hat{d}_{pivot}\mid\mathcal{E}_{tight}] + \E[\hat{d}_{pivot}\mid\mathcal{E}_{tight}^c]\nonumber\\
    &\leq \left(1+\frac{\varepsilon}{4}\right)d_{pivot}(1 - \varepsilon n^{-6})+2(n-1)\varepsilon n^{-6} \leq  \left(1 +\frac{\varepsilon}{2}\right)d_{pivot}\label{Eq: Triangle clipped degree ub}
\end{align}

Similarly, we have
\begin{align}
    \E[\hat{d}_{pivot}] &= \E[\hat{d}_{pivot}\mid\mathcal{E}_{tight}] + \E[\hat{d}_{pivot}\mid\mathcal{E}_{tight}^c]\nonumber\\
    &\geq \left(1-\frac{\varepsilon}{4}\right)d_{pivot}(1 - \varepsilon n^{-6}) \geq  \left(1 - \frac{\varepsilon}{2}\right)d_{pivot}\label{Eq: Triangle clipped degree lb}
\end{align}

Hence, combining \cref{Eq: Triangle clipped degree lb,Eq: Triangle clipped degree ub}, we have:
$\E[\hat{d}_{pivot}] \in \left[\fbrac{1 - \frac{\varepsilon}{2}}d_{pivot}, \fbrac{1 +\frac{\varepsilon}{2}}d_{pivot}\right].$

Then, we have 
$\E[Z_\ell \mid \mathcal E_{e}  ] = \frac{y_{e}}{d_{pivot}} \E[\hat d_{pivot}] = y_{e}\fbrac{1 \pm\frac{\varepsilon}{2}}$.
It follows that $\E[Y \mid \mathcal E_{e} ] = y_{e}\fbrac{1 \pm\frac{\varepsilon}{2}}$.
Then,
$\E[\esttriangle] = m \cdot \E[Y] = m \cdot \sum_{e} \frac{1}{m}\E[Y \mid \mathcal E_{e}] = T\fbrac{1 \pm\frac{\varepsilon}{2}}$\qedhere
\end{proof}

Consider a fixed set of degree estimates $\hat d$ for the vertex set $V$ of a graph $G = (V,E)$ such that for every $ v \in V$, $(1-\varepsilon)d_v \le \hat d_v \le (1+ \varepsilon) d_v$ holds, i.e. event 
$\mathcal E_{tight}$ holds and that is true with probability at least $1 - \varepsilon/n^6$, for our choice of $k$ for the \cKMV{} estimator.
Let $v_1 \prec v_2 \prec \ldots v_n$ be an ordering of the vertices defined by the estimated degrees $\hat d$ (ties broken by vertex IDs) such that we associate every triangle $\tau = (u,v,w)$ in $G$ to the edge incident on the smallest two vertices in the ordering defined by $\hat d$. (Recall that the algorithm of ~\cite{BeraC/stacs/2017/CliqueCountingEA} defines a similar total order on vertices and triangle association to edges using true degrees.)
Note that we can use this order to orient an edge $e = \{u,v\}$
from $u$ to $v$ such that $u \prec v$.
Let $y_e$ be the number of triangles assigned to a given edge
$e \in E$ 
We need the following combinatorial lemma for proving the variance, the proof of which is given in Appendix~\ref{sec:Missing-Proofs-4Pass}.

\begin{restatable}[]{lemma}{ForwardTriangleCount}\label{lem:ForwardTriangleCount}
If the estimated degree $\hat d_v$ for every $v \in V$ is within $(1 \pm \varepsilon)$ multiplicative error from the true degree of $v$, then
$y_e \le \frac{\sqrt {2m}}{1 - \varepsilon} = O_\varepsilon(\sqrt m)$.
\end{restatable}
The following lemma analyses the variance of the estimator $\esttriangle$
of our four-pass algorithm.
\begin{lemma}\label{Lem:4PassVariance}
$\Var(\esttriangle) = O(m^{3/2} \cdot T)$.
\end{lemma}
\begin{proof}[Proof of Lemma~\ref{Lem:4PassVariance}]
    The variance in the value of $\esttriangle$ is caused by two sources. 
\begin{itemize}
    \item The first is the variance due to the randomness of the sampling process, i.e., in the edge selection and the neighbourhood sampling of the pivot vertex. Observe that these are also the sources of variance when we use exact degrees instead of estimated degrees and uniform sampling without $\ell_0$-samplers.

    \item The second is the variance caused by degree estimators and $\ell_0$-samplers.
\end{itemize}
Let $\hat d$ be a vector of all the degree estimates of a given run of the algorithm.
Using the law of total variance, we can decompose $\Var(\esttriangle)$ as:
\begin{equation}\label{eq:TotalVaraince4Pass}
      \Var(\esttriangle) = \E_{\hat d}\left[\Var\left(\esttriangle \mid \hat d \right)\right] + \Var_{\hat d} \left(\E[\esttriangle \mid \hat d]\right).
\end{equation}
The first term is the expected variance of $\esttriangle$ conditioned on a fixed vector of all the degree estimates in a single run of the algorithm.
The second term accounts for the variance in the expected estimate $\E[\esttriangle \mid \hat d]$ caused by the degree estimates.
We will bound these two separately. First, we bound $\E_{\hat d}\left[\Var\left(\esttriangle \mid \hat d\right )\right]$ in the following claim, the 
proof of which is given in Appendix~\ref{sec:Missing-Proofs-4Pass}.

 \begin{restatable}[]{claim}{VarSampling}\label{Claim: Variance Due to sampling}
 There exists a constant $C'$ such that
 $ \E_{\hat d}\left[\Var\left(\esttriangle \mid \hat d\right )\right] \le  C' \cdot m^{3/2} T$.
 
 \end{restatable}

 Next, we bound the second term of~\eqref{eq:TotalVaraince4Pass}, 
$\Var_{\hat d} \left(\E[\esttriangle \mid \hat d]\right)$.
Here, we have to bound the variance of conditional expectation of output (of one run) due to the randomness in the degree estimators. We prove the following claim:

\begin{claim}\label{claim:Variance Due to Sketches}
There exists a constant $C''$ such that 
$\Var_{\hat d}\left (\E[\esttriangle \mid \hat d]\right) \le
            C'' \cdot m^{3/2} \cdot T$.

\end{claim}

\begin{proof}[Proof of Claim~\ref{claim:Variance Due to Sketches}]

\begin{align}\label{Eq:expectedEstimateGivenApprxDegrees}
    \E[\esttriangle \mid \hat d] = m \E_e[Y \mid \hat d] = 
     m \frac{1}{m} \cdot  \sum_{e} y_e \cdot \frac{\hat d_{pivot}}{d_{pivot}} 
\end{align}
Recall that the event $\mathcal E_{tight}$ occurs when for every $v \in V$, its estimated degree 
$\hat d_v$ (as computed by the estimator of Theorem~\ref{thm:F_0BJKST Clipped}) is such that
$(1 - \varepsilon) d_v \le \hat d_v \le (1 + \varepsilon) d_v$.

Therefore, conditioned on the event $\mathcal E_{tight}$, the following holds
\begin{equation}\label{eq:ratio-deviation-good}
       \left \lvert \frac{\hat d_{pivot}}{d_{pivot} } - 1  \right \rvert  \le \varepsilon.
\end{equation}
Combining~\eqref{Eq:expectedEstimateGivenApprxDegrees} and~\eqref{eq:ratio-deviation-good}, we get that conditioned on the event $\mathcal E_{tight}$
\begin{align*}
         \lvert \E[\esttriangle \mid \hat d] - T \rvert &= \left \lvert  \sum_e y_e \left ( \frac{\hat d_{pivot}}{d_{pivot}} - 1\right) \right \rvert \le \varepsilon \sum_e y_e 
     = \varepsilon T.
\end{align*}
Thus, on $\mathcal E_{tight}$,  
     $\left(\E[\esttriangle \mid \hat d] - T \right)^2 \le \varepsilon^2 T^2$.
On the complementary event $\mathcal E_{tight}^c$,
clipping ensures
$ 0 \le \E[\esttriangle \mid \hat d] \le 2m (n-1)$,
and hence, 
$\lvert \E[\esttriangle \mid \hat d] - T \rvert \le 2m (n-1) + T$.
In our invocation of Theorem~\ref{thm:F_0BJKST Clipped}, we pick $k = O(\varepsilon^{-2}(\log n + \log (1/\varepsilon)))$ such that $\Pr(\mathcal E_{tight}^c) \le \varepsilon n^{-6} \le n^{-6}$.
Combining the contribution of $\mathcal E_{tight}$ and
its complementary event $\mathcal E_{tight}^c$, we get that
\begin{align}\label{eq:expected-difference}
        \E_{\hat d}\left [ \left (\E[\esttriangle \mid \hat d] - T \right)^2  \right] 
        &\le 
        \varepsilon^2 T^2 \Pr(\mathcal E_{tight}) + 
        \left( 2m(n-1) + T \right)^2 \Pr(\mathcal E_{tight}^c)\nonumber \\
        &\le \varepsilon^2 T^2 \cdot 1 + \left( 2m(n-1) + T \right)^2 \cdot n^{-6}.
\end{align}
Finally,
$\Var_{\hat d} \left( \E[\esttriangle \mid \hat d]  \right) \le \E_{\hat d} \left [ \left( \E[\esttriangle \mid \hat d] - T \right)^2\right]$.

Using the standard bound $T = O(m^{3/2})$ (see~\cite{EdenRS/SODA/2020/LowArboricityCliques}), the first term of the right-hand-side of Equation~\eqref{eq:expected-difference} is $O(m^{3/2} T)$.
Next, we bound the second term $(2m(n-1) + T)^2 n^{-6}$. We start with the standard bound $T \le m^{3/2}$, then, we bound $m \le n^2$.
So $T \le m n $. Also, $2m(n-1) \le 2m n$.
Thus $2m(n-1) + T \le 3m n $.

Therefore, 
$(2m (n-1) + T)^2 n^{-6} \le 9m^2 n^2 \cdot n^{-6} \le 9$.
So the second term of the right-hand-side of Equation~\eqref{eq:expected-difference} is at most $O(1)$.
Therefore, 
$\Var_{\hat d} \left( \E[\esttriangle \mid \hat d]  \right) \le C'' m^{3/2} T$.\qedhere

\end{proof}
Thus, the lemma follows by Claim~\ref{Claim: Variance Due to sampling} and Claim~\ref{claim:Variance Due to Sketches}.
\end{proof}

We next analyze the space complexity of our algorithm, with the proof of the following lemma deferred to Appendix~\ref{sec:Missing-Proofs-4Pass}.


\begin{restatable}[]{lemma}{spacecompfourpass}\label{lem:spacecomp4pass}
  A single run of Algorithm~\ref{Algo: Triangle with 4-pass multi-occurrence} uses
$
    \widetilde{O}\left(\varepsilon^{-2}\log \frac{1}{\varepsilon}\right)
$
bits of space in expectation up to polylogarithmic factors required by the $\ell_0$ samplers.
\end{restatable}
Next, combining \Cref{lem:MoM} with the expected space complexity of a single run and Lemmas~\ref{Lem:4Pass-Expected} and~\ref{Lem:4PassVariance} for the expectation and variance of $\esttriangle$ we get the following theorem:
\begin{theorem}\label{Thm:4PassMain}
    There exists a $4$-pass algorithm that given an $n$-vertex $m$-edge graph presented as a graph stream that allows duplicates outputs an $(\varepsilon, \delta)$-estimate $\esttriangle$ of the number of distinct triangles 
    using $\widetilde O\left((m^{3/2}/T) \frac{\log 1/\delta}{\varepsilon^4} \right)$ space.
\end{theorem}

\begin{remark}
We note that a similar algorithm and analysis can be used to get a 4-pass algorithm for $k$- cliques with space complexity  $\widetilde O\left((m^{k/2}/T) \frac{\log 1/\delta}{\varepsilon^4} \right)$ where $T$ is the number of $k$-cliques. We give the details in the Appendix~\ref{AppSec: Cliques}. 
\end{remark}


\section{Conclusions} 

We studied triangle counting in graph streams under two natural 
generalizations of the classical single-edge arrival model: the 
\emph{repeated-edge arrival model}, where the same edge may appear 
multiple times, and the \emph{Right to be Forgotten Graph Streaming 
(RFGS) model}, where a forget operation $\perp$ retroactively removes 
all occurrences of an edge from the stream. Our main contribution is 
that, despite these generalizations, it is possible to 
match the space complexity of the state-of-the art optimal single-edge arrival algorithm 
of Jayaram and Kallaugher~\cite{JK21}. While, this algorithm introduces biases in new models, we use a \emph{last-edge} 
technique to handle this: by attributing each triangle to the final occurrence of its 
last edge in the stream, and removing triangles from the sample 
whenever one of their edges reappears, we ensure each triangle is 
counted exactly once regardless of how many times its edges repeat. We also designed space-optimal multi-pass algorithms for triangle counting in the repeated edge arrival model and and extend that to counting general cliques. Designing space optimal single-pass algorithm for counting cliques and other structures in these new models is an open question.






\phantomsection
\addcontentsline{toc}{section}{\refname}
\bibliographystyle{abbrvnat}
\bibliography{refs} 

\appendix
\section*{Appendix}\label{sec:appendix}
\setcounter{section}{0}

\section{Triangle Counting Algorithm in the RFGS Model }\label{Appendix: ForgetCode}

\Cref{Algo: Triangle with Forget Model with Bounds} contains the pseudocode for the single-pass algorithm in the RFGS model.

\begin{algorithm}[ht]
    \caption{Triangle Counting - RFGS Model}\label{Algo: Triangle with Forget Model with Bounds}
    \begin{algorithmic}[1]
        \Require $\vertexcount, \numtlb, \vsenseub, \esenseub$ such that $\numtlb \leq \numtriangle$, $\vsenseub \geq \vertexsensitivity$, and $\esenseub \geq \edgesensitivity$
        \State $\vsamprate \gets \frac{\vsenseub}{\numtlb}$
        \State $\esamprate \gets \frac{\esenseub}{\vsenseub}$ if $\esenseub \geq \sqrt{\vsenseub}$, and $\frac{1}{\sqrt{\vsenseub}}$ otherwise
        \State $\tsamprate \gets \frac{\vsenseub}{\esenseub^2}$ if $\esenseub \geq \sqrt{\vsenseub}$, and $1$ otherwise

        \State Initiate $4$-wise independent hash function $\func{\edgehash}{[\vertexcount]\times[\vertexcount]}{\tbrac{\floor{\frac{1}{\esamprate}}}}$ 
        \State Initiate $2$-wise independent hash function $\func{\vertexhash}{[\vertexcount]}{\tbrac{\floor{\frac{1}{\vsamprate}}}}$ 
        \State $\samplededges \gets \emptyset$
        \State $\sampledtriangles \gets \emptyset$
        \For{each update $\langle(\altvertex,\vertex),\Delta\rangle$}
            \If{$\Delta = +$}
                \For{$\triangle \in \sampledtriangles$}\label{Loop:Cleaning Forget}
                    \If{$(\altvertex,\vertex) \in \triangle$}
                        \State $\sampledtriangles \gets \sampledtriangles\setminus\triangle$
                    \EndIf
                \EndFor
                \For{$w \in \vertexset$}\label{Loop:Triangle Samp Forget}
                    \If{$(w,\altvertex),(w,\vertex) \in \samplededges$ \textbf{and} $\vertexhash(w) = 1$}
                        \State $\sampledtriangles \gets \sampledtriangles \cup \sbrac{w,\altvertex,\vertex} $ w.p. $\tsamprate$
                    \EndIf
                \EndFor
                \If{($\vertexhash(\vertex) = 1$ \textbf{or} $\vertexhash(\altvertex) = 1$) \textbf{and} $\edgehash((\altvertex,\vertex)) = 1$}\label{Line:Edge Samp Forget}
                    \State $\samplededges \gets \samplededges \cup (\altvertex,\vertex)$
                \EndIf
            \EndIf
            \If{$\Delta = \perp$}
                \For{$\triangle \in \sampledtriangles$}\label{Loop: Forgetting}
                    \If{$(\altvertex,\vertex) \in \triangle$}
                        \State $\sampledtriangles \gets \sampledtriangles\setminus\triangle$
                    \EndIf
                    \If{$(\altvertex,\vertex)$}
                        \State $\samplededges \gets \samplededges\setminus(\altvertex,\vertex)$
                    \EndIf
                \EndFor
            \EndIf
        \EndFor
        \State $\esttriangle \gets \size{\sampledtriangles}/\vsamprate\esamprate^2\tsamprate$
        \Return $\esttriangle$
    \end{algorithmic}
\end{algorithm}

\section{Missing Proofs from Section~\ref
{SubSec: Single Pass Multi Occurrence}}\label{Sec:1passextra}

\begin{proof}[Proof of Lemma~\ref{Lemma: MultiOccur Estimator Variance}]
    We start by bounding $\E[\numtriangle^2]$. Let us denote $\esttriangle_{\sbrac{x,y,z}}$ to be the random variable that is $\nicefrac{1}{\vsamprate\esamprate^2\tsamprate}$ if $\sbrac{x,y,z} \in \sampledtriangles$ at the end of the stream. We do the analysis on a case by case basis of the dependencies of triangles $\sbrac{u,v,w}$ and $\sbrac{x,y,z}$, and their corresponding variables $\esttriangle_{\sbrac{u,v,w}}$ and $\esttriangle_{\sbrac{x,y,z}}$. Without loss of generality, assume $\fbrac{v,w}$, and $\fbrac{y,z}$ to be the last edges of the triangles in the stream.

    \textbf{Case 1: $\sbrac{u,v,w} = \sbrac{x,y,z}$ :} The triangles are same. Hence, $\esttriangle_{\sbrac{u,v,w}} = \esttriangle_{\sbrac{x,y,z}} = \frac{1}{\vsamprate\esamprate^2\tsamprate}$ with probability $\vsamprate\esamprate^2\tsamprate$, and $0$, otherwise. Then, $\E[\esttriangle_{\sbrac{u,v,w}}\esttriangle_{\sbrac{x,y,z}}] = \frac{1}{\vsamprate\esamprate^2\tsamprate}$.

    \textbf{Case 2: $\size{\sbrac{(u,v),(u,w)}\cap\sbrac{(x,y),(x,z)}} = 1$ :} Here, the covariance is highest if $u = x$. In this case, $u = x$ has to satisfy $\vertexhash(u) = 1$, $\edgehash$ has to evaluate to $1$ for the three unique edges in $\sbrac{(u,v),(u,w)}\cup\sbrac{(x,y),(y,z)}$, and both triangles sampled with probability $\tsamprate$ for both the triangles to be in $\sampledtriangles$. Then, as $\edgehash$ is $4$-wise independent, we have:
$\E[\esttriangle_{\sbrac{u,v,w}}\esttriangle_{\sbrac{x,y,z}}] = \frac{1}{\vsamprate^2\esamprate^4\tsamprate^2}\vsamprate\esamprate^3\tsamprate^2 = \frac{1}{\vsamprate\esamprate}$.

    \textbf{Case 3: $\sbrac{(u,v),(u,w)}\cap\sbrac{(x,y),(x,z)} = \emptyset$, but $u = x$ :} Here, only the vertex $u = x$ is the only dependent component of the sampling. Then, for both the triangles to be in $\sampledtriangles$, the vertex $u = x$ has to satisfy $\vertexhash(u) = 1$, $\edgehash$ must evaluate to $1$ for all the $4$ edges, and both triangles have to be sampled with probability $\tsamprate$. Then, as $\edgehash$ is $4$-wise independent, we have: 
$\E[\esttriangle_{\sbrac{u,v,w}}\esttriangle_{\sbrac{x,y,z}}] = \frac{1}{\vsamprate^2\esamprate^4\tsamprate^2}\vsamprate\esamprate^4\tsamprate^2 = \frac{1}{\vsamprate}$.

    \textbf{Case 4: $\sbrac{u,v,w}\cap\sbrac{x,y,z} = \emptyset$} Here, the triangles are completely disjoint. Hence, we have:
    \begin{align*}
    \E[\esttriangle_{\sbrac{u,v,w}}\esttriangle_{\sbrac{x,y,z}}] = \E[\esttriangle_{\sbrac{u,v,w}}]\E[\esttriangle_{\sbrac{x,y,z}}] = 1
    \end{align*}

    Now, we bound $\E[\esttriangle^2]$ through these invididual values.
    
    \begin{align*}
        \E[\esttriangle^2] &= \sum_{\sbrac{u,v,w} \in \triangleset} \sum_{\sbrac{x,y,z} \in \triangleset} \E[\esttriangle_{u,v,w}\esttriangle_{x,y,z}]\\
        &= \sum_{\sbrac{u,v,w} \in \triangleset} \esttriangle_{\sbrac{u,v,w}^2} + \sum_{\sbrac{u,v,w} \in \triangleset} \Bigg(\sum_{\substack{\sbrac{x,y,z} \in \triangleset\\\size{\sbrac{(u,v),(u,w)}\cap\sbrac{(x,y),(y,z)}} = 1}}\E[\esttriangle_{\sbrac{u,v,w}}\esttriangle_{\sbrac{x,y,z}}]+\\
        &\sum_{\substack{\sbrac{x,y,z} \in \triangleset\\\sbrac{(u,v),(u,w)}\cap\sbrac{(x,y),(y,z)} = \emptyset\\u = x}}\E[\esttriangle_{\sbrac{u,v,w}}\esttriangle_{\sbrac{x,y,z}}]+\sum_{\substack{\sbrac{x,y,z} \in \triangleset\\\sbrac{u,v,w}\cap\sbrac{x,y,z} = \emptyset}}\E[\esttriangle_{\sbrac{u,v,w}}\esttriangle_{\sbrac{x,y,z}}]\Bigg)\\
        &\leq\frac{\numtriangle}{\vsamprate\esamprate^2\tsamprate}+\frac{\numtriangle\edgesensitivity}{\vsamprate\esamprate}+\frac{\numtriangle\vertexsensitivity}{\vsamprate} + \numtriangle^2
    \end{align*}
    The fact that $\Var[\esttriangle] = \E[\esttriangle^2] - \E^2[\esttriangle]$, and~\Cref{Lemma: MultiOccur Unbiased Estimator} completes the proof.
\end{proof}

\begin{proof}[Proof of Corollary~\ref{Cor: MultiOccur Variance Control}]
    We bound the terms for the two cases of $\esenseub$ separately. 

    \textbf{Case 1: $\esenseub \leq \sqrt{\vsenseub}$ :} 
    We bound the terms $\vsamprate\esamprate^2$, $\nicefrac{\vsamprate\esamprate}{\edgesensitivity}$, and $\nicefrac{\vsamprate}{\vertexsensitivity}$. Here $\vsamprate = \frac{\vsenseub}{\numtlb}$, $\esamprate = \frac{1}{\sqrt{\vsenseub}}$, $\tsamprate = 1$:

    \begin{align*}
        \vsamprate\esamprate^2\tsamprate &= \vsamprate\esamprate^2 = \frac{\vsenseub}{\numtlb}\frac{1}{\vsenseub} = \frac{1}{\numtlb} \\
        \frac{\vsamprate\esamprate}{\edgesensitivity} &\geq \frac{\vsenseub}{\numtlb}\frac{1}{\sqrt{\vsenseub}}\frac{1}{\esenseub} = \frac{\sqrt{\vsenseub}}{\numtlb\esenseub} \geq \frac{1}{\numtlb}\\
        \frac{\vsamprate}{\vertexsensitivity} &\geq \frac{\vsenseub}{\numtlb}\frac{1}{{\vsenseub}} = \frac{1}{\numtlb}
    \end{align*}

    \textbf{Case 2: $\esenseub \geq \sqrt{\vsenseub}$ :} 
    We bound the terms one by one. Here $\vsamprate = \frac{\vsenseub}{\numtlb}$, $\esamprate = \frac{\esenseub}{{\vsenseub}}$, $\tsamprate = \frac{\vsenseub}{\esenseub^2}$:

    \begin{align*}
        \vsamprate\esamprate^2\tsamprate &= \frac{\vsenseub}{\numtlb}\frac{\esenseub^2}{\vsenseub^2}\frac{\vsenseub}{\esenseub^2} = \frac{1}{\numtlb}\\
        \frac{\vsamprate\esamprate}{\edgesensitivity} &\geq \frac{\vsenseub}{\numtlb}\frac{\esenseub}{\vsenseub}\frac{1}{\esenseub} = \frac{1}{\numtlb}\\
        \frac{\vsamprate}{\vertexsensitivity} &\geq \frac{\vsenseub}{\numtlb}\frac{1}{\vsenseub} = \frac{1}{\numtlb}
    \end{align*}

    Combining the results, we obtain $\Var[\esttriangle] \leq \numtriangle^2 + 3\numtriangle\numtlb \leq 4\numtriangle^2$.
\end{proof}

\section{Lower Bound}\label{AppSec: LB}

In this section, we establish our lower bound for the RFGS model. We establish our result using a reduction to $\bF_0$ estimation in the RFDS model. \Cref{Thm: RFGS Lower bound} follows from \Cref{Lem: RFGS Tmax LB,Lem: Vertex Sensitivity Lower Bound,Lem: Edge Sensitivity Lower Bound}. We will use the following result due to Pavan  et al.~\cite{PavanCVM/PODS/2024/RFDSIntro} on space complexity of $\bF_0$ estimation.

\begin{lemma}\label{Thm: F_0 Lower Bound}
    Given a stream $\bS$ of elements from $\sU = \sbrac{\sequence{u}{n}}$, denote $\bF_0(\bS)$ to be the number of distinct elements in the stream. Let us consider the stream $\bS'$ obtained by removing all forget operations in $\bS$. Then, any algorithm that obtains a $(\nicefrac{1}{10},\nicefrac{1}{10})$-estimate of $\bF_0(S)$ required $\bigomega{\frac{\bF_0(\bS')}{\bF_0(\bS)}}$ space.
\end{lemma}

\subsection{Triangle Dependent Lower Bound}

In this section, we show that the dependence of the space complexity on $\numtriangle_{\max{}}$ is necessary in the RFGS model. 


\begin{lemma}\label{Lem: RFGS Tmax LB}
    Any algorithm that obtains a $(\nicefrac{1}{3},\nicefrac{1}{3})$-estimate of the number of triangles in the RFGS model must use $\bigomega{\frac{\numtriangle_{\max{}}}{\numtriangle}}$ space.
\end{lemma}

\begin{proof}[Proof of Lemma~\ref{Lem: RFGS Tmax LB}]
    We prove by contradiction. Let us assume there exists a streaming algorithm $\sA$ that obtains a $(\nicefrac{1}{3},\nicefrac{1}{3})$-estimate of the number of triangles using $\smallo{\frac{\numtriangle_{\max{}}}{\numtriangle}}$ space.

    Now, we use $\sA$ to design an algorithm for estimating $\bF_0$. For each element $u_i$ in $\sU$, we consider three vertices $u^1_i,u_i^2$, and $u_i^3$. Let us consider a stream $\bS = \sequence{s}{m}$. We now outline how we generate a graph stream $\bS_\graph$ for a graph $\graph_{\bS} = (\vertexset_{\bS},\edgeset_{\bS})$ with vertex set 
    $\vertexset_{\bS{}} = \sbrac{u^1_1,u_1^2,u_1^3,\ldots,u_n^1,u_n^2,u_n^3}$:

    \begin{algorithm}
    \hrulefill
        \begin{algorithmic}
            \For{$s_j = (u_i,\Delta) \in \bS$}
                \If{$\Delta = +$}    
                    \State Add updates $((u_i^1,u_i^2),+),((u_i^2,u_i^3),+),((u_i^3,u_i^1),+)$ to the graph stream $\bS_\graph$.
                \EndIf
                \If{$\Delta = \perp$}
                    \State Add updates $((u_i^1,u_i^2),\perp),((u_i^2,u_i^3),\perp),((u_i^3,u_i^1),\perp)$ to the graph stream $\bS_\graph$.
                \EndIf
            \EndFor
        \end{algorithmic}
    \end{algorithm}

    Note that the number of remaining triangles at the end of the graph stream is the number of surviving elements in the original stream. Hence, we have $\numtriangle\fbrac{\bS_\graph} = \bF_0(\bS)$ where $\numtriangle\fbrac{\bS_\graph}$ denotes the number of triangles contained in the stream $\bS_\graph$. Let us also denote $\numtriangle_{\max}(\bS_{\graph})$ to be the maximum number of triangle in the graph corresponding to the stream $\bS_{\graph}$ at any point in the stream.
    
    Let $\bS'_\graph$ denote the graph stream without all the forget operations, and $\bS'$ denote the stream obtained by removing all forget operations from $\bS$.   By a similar argument as before, we have $\numtriangle\fbrac{\bS'_\graph} = \bF_0(\bS')$. Also note that $\numtriangle\fbrac{\bS'_\graph} = \numtriangle_{\max}(\bS_{\graph})$, and consequently, $\numtriangle_{\max} \leq \bF_0(\bS')$. 

    Since $\sA$ is a $(1/3, 1/3)$ approximation of $T(\bs_\graph)$, we have a $(1/3,1/3)$ approximation of $\bF_0(\bS)$ and this algorith uses space 
    \begin{align*}
        \smallo{\frac{\numtriangle_{\max}(\bS_{\graph})}{\numtriangle(\bS_{\graph})}} = \smallo{\frac{\bF_0(\bS')}{\bF_0(\bS)}}
    \end{align*}
    Hence, by \Cref{Thm: F_0 Lower Bound}, we have a contradiction.
\end{proof}

\subsection{Edge Sensitivity and Vertex Sensitivity Lower Bounds}

\begin{lemma}\label{Lem: Edge Sensitivity Lower Bound}
    Any algorithm that obtains a $(\nicefrac{1}{3},\nicefrac{1}{3})$-estimate of the number of triangles in the RFGS model must use $\bigomega{\frac{m_{\max{}}\edgesensitivity}{\numtriangle}}$ space.
\end{lemma}

\begin{proof}[Proof of Lemma~\ref{Lem: Edge Sensitivity Lower Bound}]
    We prove by contradiction. Let us assume there exists a streaming algorithm $\sA$ that obtains a $(\nicefrac{1}{3},\nicefrac{1}{3})$-estimate of the number of triangles using $\smallo{\frac{m_{\max{}}\edgesensitivity}{\numtriangle}}$ space.

    Now, we use $\sA$ to design an algorithm for estimating $\bF_0$.  Let $\bS = s_1, s_2, \cdots, s_m$ be a stream (with forget operations) over a universe $\sU$ for which we wish to estimate $F_0$. We create a graph steram $\bS_\graph$ from $\bS$. For each element $u_i$ in $\sU$, we associate a set of vertices $\sbrac{\sequence{v^i}{K}}$, for a suitable choice of $K$. For concreteness sake we can take $K$ to $\log |\sU|$. For every $u_i$ that appears in stream $\bS$, the stream $\bS_\graph$ has a clique over the vertex set $\sbrac{\sequence{v^i}{K}}$.
    
    More formally the graph stream $\bS_\graph$ represents a a graph $\graph_{\bS} = (\vertexset_{\bS},\edgeset_{\bS})$ with vertex set 
    $\vertexset_{\bS{}} = \cup_{i \in [n]} \sbrac{\sequence{v^i}{K}}$ and set of edges $\edgeset_{\bS{}}$ created as follows:

    \begin{algorithm}
    \hrulefill
        \begin{algorithmic}
            \For{$s_j = (u_i,\Delta) \in \bS$}
                \If{$\Delta = +$}    
                    \State Add updates $\sbrac{\fbrac{(v^i_j,v^i_k),+}\mid 1\leq j\leq k \leq K}$ to the graph stream $\bS_\graph$.
                \EndIf
                \If{$\Delta = \perp$}
                    \State Add updates $\sbrac{\fbrac{(v^i_j,v^i_k),\perp}\mid 1\leq j\leq k \leq K}$ to the graph stream $\bS_\graph$.
                \EndIf
            \EndFor
        \end{algorithmic}
    \end{algorithm}
    Observe that each time an element appears in the original stream $\bS$, it contributes 
        $\binom{K}{2}$  edges and $\binom{K}{3}$  triangles   
   to $\bS_\graph$. Hence, the number of triangles at the end of the graph stream is the number of elements at the end of the original stream multiplied by $\binom{K}{3}$, i.e.
    \begin{align*}
        \numtriangle(\bS_\graph) = \bF_0(\bS)\binom{K}{3}
    \end{align*}
    Hence, if an algorithm can estimate the number of triangles in the graph stream $\bS_\graph$, then it can be used to estimate $\bF_0(\bS)$.  Next, observe that the number of elements present at any point in the original stream is at most $\bF_0(\bS')$, where $\bS'$ is the stream obtained by removing forget operations from $\bS$.  Hence,
    \begin{align*}
        m_{\max{}} \leq \bF_0(\bS')\binom{K}{2}
    \end{align*}
    Finally note that for the stream $\bS_\graph$, $\edgesensitivity$ is exactly $K-2$. Now, note that if $\sA$ can estimate the number of triangles using $\smallo{\frac{m_{\max{}}\edgesensitivity}{\numtriangle}}$ space, then we can estimate $\bF_0(S)$ using space
    \begin{align*}
        \smallo{\frac{m_{\max{}}\edgesensitivity}{\numtriangle}} \leq \smallo{\frac{\bF_0(\bS')\binom{K}{2}K}{\bF_0(\bS)\binom{K}{3}}} = \smallo{\frac{\bF_0(\bS')}{\bF_0(\bS)}}
    \end{align*}
    Hence, a contradiction.
\end{proof}

We can use the same ideas to establish the following.

\begin{lemma}\label{Lem: Vertex Sensitivity Lower Bound}
    Any algorithm that obtains a $(\nicefrac{1}{3},\nicefrac{1}{3})$-estimate of the number of triangles in the RFGS model must use $\bigomega{\frac{m_{\max{}}\sqrt{\vertexsensitivity}}{\numtriangle}}$ space.
\end{lemma}

The proof proceeds exactly as before. We note that for the graph stream $\bS_\graph$, $\vertexsensitivity$ equals $\binom{K-1}{2}$. Suppose $\sA$ is an algorithm that uses $\smallo{\frac{m_{\max{}}\sqrt{\vertexsensitivity}}{\numtriangle}}$ space and return $(1/3, 1/3)$-estimate the number of triangles in $\sG_\graph$. The same algorithm returns a $(1/3, 1/3)$-estimate of $F_0(\bS)$. The space used by the algorithm is 
 \begin{align*}
        \smallo{\frac{m_{\max{}}\sqrt{\vertexsensitivity}}{\numtriangle}} \leq \smallo{\frac{\bF_0(\bS')\binom{K}{2}\sqrt{\binom{K}{2}}}{\bF_0(\bS)\binom{K}{3}}} = \smallo{\frac{\bF_0(\bS')}{\bF_0(\bS)}}
    \end{align*}
    which is a contradiction.

\section{Two-Pass Algorithm}\label{sec:twoPass}
Our two-pass algorithm for the \textsf{scarce-regime} is based on treating light edges that participate in a small number of triangles separately from heavy edges that participate in a large number of triangles. We first start by sketching a simple one-pass algorithm based on~\cite{JhaSP/TKDD/2015/TriangleCOuntingStream} to motivate the need to treat light and heavy edges separately. (We note here that this point has been elaborated by McGregor et al. in the preamble to Section $3.1$ of~\cite{McGregorVV/pods/2016/BetterTriangleCountingStreams}.
We summarize it here for completeness.)
\begin{itemize}
    \item Initiate a counter $Z$.
    \item Sample each edge independently with probability $p$.
    \item If an edge completes a triangle with previously sampled edges, update $Z \gets Z + 1$.
    \item Return as estimate $\esttriangle \gets Z/p^2$.
\end{itemize}

Observe that $\esttriangle$ is an unbiased estimator of $\numtriangle$ as $\E[Z] = \numtriangle p^2$, and hence $\E[\esttriangle] = \numtriangle$. The problem however lies in controlling the variance with $\Var[Z] = \numtriangle p^2 + \numtriangle \edgesensitivity p^3$ (see Lemma~$7$ of~\cite{McGregorVV/pods/2016/BetterTriangleCountingStreams} for a proof). In the single-pass setting, the tight bounds establish that we can not get rid of the $\edgesensitivity$ term. However, in the multi-pass model, we can deal with it by categorizing  the edges as heavy and light.

Next, we give a brief overview of the two-pass algorithm of~\cite{McGregorVV/pods/2016/BetterTriangleCountingStreams} in the standard graph stream model (with just one occurrence of every edge) and then describe how we extend this to the graph streams that allow repeated occurrences. We refer to the two-pass algorithm described in Section $3.1$ of~\cite{McGregorVV/pods/2016/BetterTriangleCountingStreams} as MVV henceforth.

The MVV algorithm builds an oracle in the first pass which is later invoked via an operation $\textsf{oracle} (e)$ to characterize a given
edge $e$ as light or heavy.  Specifically, the oracle
consists of a subset $E_{oracle} \subseteq E$ of edges of the input graph and is created by the following two-step process:
\begin{itemize}
    \item Sample every vertex $v \in V$ with probability $p = \tilde \Theta_{\varepsilon, n} \left(\frac{1}{\sqrt {\numtlb} }\right) $ and add the sampled vertices to a set $V_{oracle} \subseteq V$.
    \item For every edge $e = \{u,v\}$ in the stream, add $e$ to $E_{oracle}$, if any of its endpoints is contained in $V_{oracle}$.
\end{itemize}
Let $x_e$ be the number of triangles (in $G$) an edge $e \in E$ participates in.
Given the oracle edges, every edge $e \in E$ can be characterized as light or heavy by counting the number of triangles it forms with the 
oracle edges. Specifically, for a given edge $e = \{u,v\}$, we can compute
$\overline{x}_e = |\{w \in V_{oracle} \mid \{u,w\} \in E_{oracle} ~\land~ \{v,w\} \in E_{oracle}  \}|$ by only looking at the \textsf{oracle} edges. The following simple definition for the operation $\textsf{oracle}(.)$ then can be used to characterize the edges as light or heavy.
\[
\textsf{oracle}(e) = 
\begin{cases}
L, & \text{if } \overline x_e < 2p \sqrt {\numtlb} \\
H,   & \text{if } \overline x_e \ge 2p \sqrt {\numtlb}.
\end{cases}
\]
Observe that $\overline x_e$ is a binomial random variable $\overline x_e \sim \textbf{Bin}(x_e, p)$, where $p = \tilde \Theta_{\varepsilon, n} (1/\sqrt{\numtlb})$.
The following lemma stated in~\cite{McGregorVV/pods/2016/BetterTriangleCountingStreams} then follows by applying Chernoff bound to the random variable $\overline x_e$.
\begin{lemma}[Lemma~$9$ of~\cite{McGregorVV/pods/2016/BetterTriangleCountingStreams}]\label{lem:twoPassOracle}
For a suitable constant $c$, With  probability at least $1 - n^{-c}$, for all $e = \{u, v\}$, 
$\textsf{oracle}(e) = L \text{~implies~} x_e \le 2\sqrt T$ and $\textsf{oracle}(e) = H \text{~implies~} x_e > 2\sqrt T$.
\end{lemma}
In addition to computing the edges of the oracle, MVV algorithm also samples a random subset $E_{sampled}$ of edges in the first pass by including every new edge on the stream to $E_{sampled}$ with probability $p$. Recall that $p = \tilde \Theta_{\varepsilon, n}(1/\sqrt{\numtlb})$.

Notice that once the vertices and edges that form the oracle have been 
extracted, the output of operation $\textsf{oracle} (e)$ is fixed for every edge $e \in E$. Therefore, one can run the operation $\textsf{oracle} (e)$ on every edge of $E_{sampled} \cup E_{oracle}$ offline at the end of the first pass. This lets us partition the set $E_{sampled}$ into two sets $E^L_{sampled}$
and $E^H_{sampled}$ corresponding to the light and heavy edges of $E_{sampled}$ respectively. Similarly we can also partition the set $E_{oracle}$ into two sets $E^L_{oracle}$ and $E^H_{oracle}$.
This concludes the description of MVV up to the beginning of the second pass.

In the second pass, the algorithm maintains two counters $A_L$ and $A_H$.
The counter $A_L$ is used to keep a running estimate
of the number of triangles containing only light edges and the counter $A_H$ does the same for
triangles containing at least one heave edge.
For each new edge that is light MVV
increments $A_L$ by a third of the number of triangles it forms
with the edges from $E^L_{sampled}$. In expectation $A_L$ is
$p^2$ times the total number of triangles involving three light edges. 
For each new heavy edge $e$,  the algorithm uses the oracle to estimate the
number of triangles with $i$ heavy edges, $i \in {1, 2, 3}$,  that $e$ participates in.
Incrementing $A_H$ by the sum of these
estimates will count a triangle with $i$ heavy edges $i$ times. But
by dividing the estimated contribution of $e$ to the number of triangles with $i$ heavy edges
by $i$, the algorithm  ensures that $A_H$ is an estimate
of the total number of triangles with at least one heavy edge.
This completes the description of the MVV algorithm.

Next, we see how we adapt the MVV algorithm to work with graph streams that allow repeated 
occurrences of edges. During the first pass the MVV algorithm builds the oracle and samples a
subset $E_{sampled}$ of edges from the stream. For building the oracle we sample the vertex set 
$V_{oracle}$ of the oracle a priori as a pre-processing step before the first pass through the stream. Note that this is not a problem since we assume that
vertices have IDs which are numbers from 
$1$ to $n$ and all we have to do is to sample an appropriately sized  subset of $\{1,2, \ldots, n \}$.
During the first pass, on arrival of every new edge $e = \{u, v\}$, we first check if either $u$ or $v$ belongs to $V_{oracle}$ and if so, we add the edge $e$ to $E_{oracle}$, if not already there.
This builds the oracle. Sampling the vertex set a priori ensures that
for any edge $e \in E$, every wedge $W_e$ in the set of wedges that $e$ completes to form a triangle is sampled independently (with probability $p$) into oracle.
In addition, on arrival of every new edge $e$, we sample it into the set
$E_{sampled}$ using a $3$-wise independent hash function. This ensures that every occurrence of a given edge $e$ on the stream is treated consistently. For a given  edge $e \in E$ such that edge $e$ is sampled by the hash function is added to the set 
$E_{sampled}$ on its first occurrence and remains there when encountered again.

In the second pass, we maintain four counters.
These are $F_L$ which estimates the number of triangles with only light edges, $F^{(i)}_H$, $i \in \{1,2,3\}$, each of which estimates the number of triangles with $i$ heavy edges.
To take care of over-counting of triangles due to multiple occurrences of a given edge $e$, instead of simply counting the number of triangles $e$ forms with the edges from $E_{sampled}$ and $E_{oracle}$, we use an distinct elements estimator $\cKMV_k$ instead of a counter to keep track of the number of distinct triangles.
In essence, each of $F^L$ and $F^{(i)}_H$ is an $F_0$ counter.
On arrival of a light edge $e= \{u, v\}$ on the stream, we check if it forms a triangle with  the edges from $E^L_{sampled}$
and update the corresponding $\cKMV_k$ estimator.
On arrival of a heavy edge $e= \{u, v\}$ on the stream, we check if it forms a triangle containing $i$ heavy edges with  the edges from $E_{oracle}$
and update the $\cKMV_k$ estimator for $F^{(i)}_H$. 


The pseudocode for our algorithm is given in Algorithm~\ref{Algo: Triangle with 2-pass multi-occurrence}.

\ 

\noindent \textbf{Analysis}:
In the rest of this section we refer to a triangle containing only light edges  as a \emph{light triangle} and a triangle containing at least one heavy edge as a \emph{heavy triangle}.
We will start by analysing the contribution of light triangles to our estimate $\esttriangle$.
\Cref{lem:twoPassOracle} stated above (for arbitrary order streams with a single occurrence of every edge)
holds for Algorithm~\ref{Algo: Triangle with 2-pass multi-occurrence}. Algorithm~\ref{Algo: Triangle with 2-pass multi-occurrence} creates the oracle by sampling vertices $V_{oracle}$ a priori before the first pass through the stream. The edge set $E_{oracle}$ of the oracle
is computed during the first pass by storing the first occurrence of every edge incident to
a vertex $v \in V_{oracle}$ and ignoring the subsequent occurrences. For a given edge $e \in E$, the above process ensures that each of the $x_e$ triangles an edge $e$ participates in $G$ is sampled independently with probability $p$. Therefore, we get that $\overline x_e \sim \textbf{Bin}(x_e, p)$ and the claim of the \Cref{lem:twoPassOracle} follows by applying Chernoff bound.

Let $\mathcal T_L$ be the set of distinct light triangles in the set $E_L = \{e \in E \mid \textsf{oracle(e)} = L \}$ and 
$T_L = \size{\mathcal T_L}$. Let $\widehat T_L$ be the number of distinct light triangles detected by the algorithm in the second pass 
(Line~\ref{lin:LightDetected}).
We have the following:
\begin{lemma}\label{Lem:2PassLightDetect}
For an error parameter $\varepsilon \ge 0$, the following holds for the clipped $F_0$ estimator used to keep track of the light triangles in Algorithm~\ref{Algo: Triangle with 2-pass multi-occurrence}
       \[  
            (3p^2 - 2p^3)(1 -\varepsilon/2)  \cdot  T_L \le \E[F_L] \le (3p^2 - 2p^3) T_L.
       \]
when we use $\cKMV_k$ estimators with  $k \ge c \log (4/\varepsilon)$, by~\Cref{thm:F_0BJKST Clipped}.
\end{lemma}
\begin{proof}[Proof of Lemma~\ref{Lem:2PassLightDetect}]
For every triangle $\tau \in \mathcal T_L$, let $X_{\tau}$
be the indicator that $\tau$ is inserted into the $\cKMV_{k}$ estimator $F_L$.
A light triangle can be detected when any one of its three edges arrives on the stream, and the other two edges have been sampled into $E_{sampled}$. Let $A_1$, $A_2$ and $A_3$ denote these three detection events. 

Since each edge is sampled with probability $p$ using a $3$-wise independent hash function, $\Pr(A_i) = p^2$, for each $i \in [3]$.

Moreover, the intersection of any two detection events requires
all three edges of $\tau$ to be sampled. Therefore,
$\Pr(A_i \cap A_j) = p^3$, for $ i \neq j$.
Similarly, $\Pr(A_1 \cap A_2 \cap A_3) = p^3$.

Since the $\cKMV_k$ estimator counts $\tau$ exactly once even if it is detected through multiple wedges, we have
\[
    \Pr(X_{\tau} = 1) = \Pr(A_1 \cup A_2 \cup A_3).
\]
By inclusion-exclusion, we have
$\Pr(X_{\tau} = 1) = 3p^2 - 3p^3 + p^3 = 3p^2 - 2p^3$.
Therefore, the expected number of distinct light triangles $D_L$ inserted into$\cKMV_k$ estimator is given by
\begin{equation}\label{eq:sampling-expect}
  \E[D_L] =  E\left [\sum_{\tau \in  T_L} X_{\tau}\right ] = (3p^2 - 2p^3)T_L.
\end{equation}
Now, condition on the sampled edge set, so that $D_L$ is fixed. 
The clipped KMV estimator of Theorem~\ref{thm:F_0BJKST Clipped} gives
$D_L(1 - 2e^{-ck/2})) \le  \E[F_L \mid D_L ] \le D_L$.
For any $k \ge (2/c) \log (4/\varepsilon)$, we get 
$D_L(1 -\varepsilon/2)) \le  \E[F_L \mid D_L ] \le D_L$.

Now, taking expectation over the sampling event and $D_L$, we get that
\begin{equation}\label{eq:F-L-expect}
    (1 -\varepsilon/2)\E[D_L] \le \E[F_L] \le \E[D_L].
\end{equation}
Combining~\eqref{eq:sampling-expect} an~\eqref{eq:F-L-expect} gives us
\[ (1 -\varepsilon/2) (3p^2 - 2p^3) T_L \le \E[F_L] \le (3p^2 - 2p^3) T_L .\qedhere
\]
\end{proof}

Next, using an argument similar to the one used in the proof of Lemma~$7$ of~\cite{McGregorVV/pods/2016/BetterTriangleCountingStreams} we show that the following holds for detection of light triangles at Line~\ref{lin:LightDetected} when sampling is done using hash functions.

\

 Let $D_L$ be the number of distinct light triangles detected by the Algorithm~\ref{Algo: Triangle with 2-pass multi-occurrence}, and let $\gamma = \sum_{e \in E_L} \binom{x_e}{2} $, where $x_e$ is the number of light triangles $e$ participates in.
Then, the following lemma holds.
\begin{lemma}\label{Lem:LightVarainceSampling}

$\Var(D_L) \le 3p^2 T_L + p^3 \gamma$.
\end{lemma}
\begin{proof}[Proof of Lemma~\ref{Lem:LightVarainceSampling}]
Let $X_1, X_2, \ldots, X_{T_L}$ be indicator random variables where $X_i =1$ indicates that the $i$th light triangle in $\mathcal T_L$ was detected by the algorithm.
Since $\cKMV_k$ estimator counts each triangle only once, we have
$D_L = \sum_{i =1}^{T_L} X_i$.
For every $i$, $\E[X_i] = 3p^2 - 2p^3$ and hence $\Var[X_i] \le \E[X_i] \le 3p^2 - 2p^3 \le 3p^2$.
Therefore $\sum_i \Var(X_i ) \le 3p^2 T_L$.
Now, consider two distinct triangles $T_i$ and $T_J$.
Their covaraince,  $\operatorname{Cov}(X_i, X_j) \le \E[X_i X_J]$ is at most the probability of both $T_i$ and $T_j$ getting detected if they share
share an edge and it is $0$, otherwise. Since we use a $3$-wise independent 
hash function to sample the edges, the probability of two triangles with a common edge getting detected is $p^3$. It follows therefore that $\operatorname{Cov}(X_i, X_j) \le p^3$.
For an edge $e$, there are $x_e(x_e -1)$ ordered pairs of distinct light triangle sharing edge $e$. 
Hence $\sum_{ i \neq j} \operatorname{Cov}(X_i, X_J) \le p^3 \cdot \sum_e x_e(x_e -1)\le p^3 \sum_e x_e^2 = p^3 \gamma$.
Therefore, 
\[\Var(D_L) \le \sum_i \Var(X_i) + \sum_{i \neq j} \operatorname{Cov} (X_i, X_j)
    \le 3p^2 T_L + p^3 \gamma.\qedhere \]
\end{proof}

\begin{lemma}\label{lem:errorProb}
   Let $q = 3p^2 - 2p^3$ and define $Y_L = F_L/q$. For a given error parameter 
   $\varepsilon_{\ell}$, and  $p = \beta \log n/\varepsilon_{\ell}^2 \sqrt{\numtlb}$ and 
   $k = \Theta(\log n/\varepsilon_{\ell}^2)$, 
   for sufficiently large constant $\beta$, the following holds
   \[
       \Pr(\lvert Y_L - T_L \rvert \le \varepsilon_{\ell} T) \ge 99/100.
   \]
\end{lemma}
\begin{proof}[Proof of Lemma~\ref{lem:errorProb}]
We decompose the error as 
\begin{align}\label{eq:error-decompose}
    \lvert Y_L - T_L \rvert = \lvert F_L/q - T_L \rvert 
    \le \lvert F_L - D_L \rvert/q + \lvert D_L/q - T_L  \rvert .
\end{align}
We first bound the sampling error $\lvert D_L/q - T_L  \rvert$.
By Lemma~\ref{Lem:LightVarainceSampling} and Chebyshev's inequality, 
\begin{align*}\label{eq:samplingError}
    \begin{aligned}
        \Pr(\lvert D_L/q - T_L  \rvert \ge \varepsilon_{\ell} T/2) &\le 
                  \frac{4 \Var(D_L)}{\varepsilon_{\ell}^2 q^2 T^2}\\
                  & \le \frac{4(3p^2 T_L + \gamma p^3)}{\varepsilon_{\ell}^2 q^2 T^2} =
                  O( [1/(\varepsilon_{\ell}^2 p^2 T)] + [1/(\varepsilon_{\ell}^2 p T)] )\\
                  & \le O( \varepsilon_{\ell}^2/(\beta^2 \log^2 n) + O(1/(\beta \log n \sqrt {\numtlb}) ).
      \end{aligned}
\end{align*}
Now, invoking Lemma~\ref{lem:twoPassOracle}, for every light edge $e$, $x_e \le 2\sqrt T$ and hence
$\gamma = \sum_{e \in E_L} \binom{x_e}{2} = O(T)$ and 
we can choose $\beta$ sufficiently large so that 
\begin{equation}\label{eq:first-error}
    \Pr(\lvert D_L/q - T_L \rvert \ge \varepsilon_{\ell} T/2) \le 1/200.
\end{equation}
Next, conditioned on $D_L$, by Theorem~\ref{thm:F_0BJKST Clipped}, 
\begin{equation}
    \Pr(\lvert F_L - D_L \rvert \ge \delta D_L \mid D_L) \le 2 e^{-c k \delta^2/2 }
\end{equation}
Set $\delta = \varepsilon_{\ell}/4$, choosing $k = \Theta(\log n/\varepsilon_{\ell}^2) $ gives
\begin{equation}\label{eq:second-error}
      \Pr(\lvert F_L - D_L \rvert \ge (\varepsilon_{\ell}/4) D_L) \le 1/200.
\end{equation}
On the event in~\eqref{eq:first-error}, $D_L/q \le T_L + (\varepsilon_{\ell} T_L)/2 \le 3T/2$,
since $\varepsilon_{\ell} \le 1$ and $T_L \le T$.
Hence, whenever both good events hold, $\lvert F_L - D_L \rvert/q \le (\varepsilon/4) (D_L /q) \le 3/8 T$.
Now, invoking~\eqref{eq:error-decompose}, we get
\[
    \lvert Y_L - T_L \rvert \le (\varepsilon_{\ell} T)/2 + (3 \varepsilon_{\ell} T) /8 \le \varepsilon_{\ell} T.
\]
Finally, by union-bound and~\eqref{eq:first-error} and~\eqref{eq:second-error}, 
$\Pr(\lvert  Y_L- T_L \rvert  \ge \varepsilon_{\ell} T) \le 1/100$.\qedhere
\end{proof}

\noindent Now, we analyze the case of heavy triangles.
Let $T_H$ be the number of distinct heavy triangles in $G$.
\begin{lemma}\label{Lem:heavyErrorFinal}
  Fix an error parameter $\varepsilon_h >0$.
  Let $T_H$ denote the number of triangles containing at least one heavy
  Then, for sufficiently large $k = \Theta(\log n/\varepsilon_h^2)$
   the following holds
  \[
      \Pr\left( \left|\frac{\sum_i F^{(i)}_H}{i p} - T_H \right| \le  \varepsilon_h T_H\right) \ge 98/100,
  \]
where $F_H^{(i)}$ is the output of $\cKMV_k$ estimator for the number of triangles with $i$ heavy edges.
\end{lemma}
\begin{proof}[Proof of Lemma~\ref{Lem:heavyErrorFinal}]
For a heavy edge $e$ and $i \in \{1,2,3\}$, let $x_e^{(i)}$ denote the number of triangles containing edge $e$ and exactly $i$ heavy edges.
For each heavy edge $e$, let $D_e^{(i)}$ be the number of distinct identities $(e, \tau)$ inserted into $F_H^{(i)}$. 
Because a triangle $\tau$ containing $e$ has a unique third vertex that witnesses the triangle through edge $e$, the identity $(e, \tau)$ is inserted exactly when the witness vertex belongs to $V_{oracle}$.
Since every vertex is sampled independently into $V_{oracle}$ with probability $p$, $D_e^{(i)} \sim \operatorname{Bin}(x_e^{(i)}, p)$.
Thus $\E[D_e^{(i)}] = p \cdot x_e^{(i)}$.

Define $D_H^{(i)} = \sum_{\textsf{oracle}(e) = H} D_e^{(i)} $.
Since $(e, \tau)$ contains the heavy edge as part of its identification, 
different heavy edges give different $F_0$ items. 
Hence, $D_H^{(i)}$ is exactly the number of distinct items represented in $F_H^{(i)}$. 
By the Chernoff bound applied to $\sum_i D_e^{(i)}$ for each edge $e$, we have
\begin{align*}
     \abs{\sum_i D_e^{(i)} - \sum_i p x_e^{(i)}} \le \sum_i (\varepsilon_h/2) p x_e^{(i)}
\end{align*}
with probability at least $99/100$
simultaneously for all heavy edges.
Since $T_H = \sum_{e :\textsf{oracle}(e) = H} x_e^{(1)} + x_e^{(2)}/2 +  x_e^{(3)}/3 $, we obtain
\begin{equation}\label{eq:SamplingError-Heavy}
     \abs{\sum_{i = 1}^{3}  D_H^{(i)}/(ip) - T_H} \le (\varepsilon_h /2) T_H.
\end{equation}

Now, consider the error due to the use of$\cKMV_k$ estimator instead of actual counters. 
Let the internal $F_0$ error parameter be $\delta = \varepsilon_h/3$.
Conditioned on the sampled sets, $D_H^{(i)}$ is the true number of distinct items presented to the respective $F_H^{(i)}$ $\cKMV_k$ estimator.
By the clipped $F_0$ guarantee, for sufficiently large $k = \Theta(\log n/\delta^2) = \Theta(\log n/\varepsilon_h^2)$, all three $\cKMV_k$ estimators simultaneously satisfy, 
\begin{equation}\label{eq:F0Error-Heavy1}
\lvert F_H^{(i)} - D_H^{(i)} \rvert \le \delta D_H^{(i)},
\end{equation}
with probability at least $99/100$.

On the intersection of sampling and $F_0$ good events, we have
\begin{align}\label{eq:FinalErrorHeavy}
\begin{aligned}
   \lvert \sum_i F_H^{(i)}/(ip) -  T_H\rvert \le  
   \lvert \sum_i (F_H^{(i)} - D_H^{(i)} )/(ip)\rvert +
   \lvert \sum_i D_H^{(i)}/(ip)  -  T_H  \rvert 
\end{aligned}
\end{align}
The second term of Equation~\eqref{eq:FinalErrorHeavy} is at most $(\varepsilon_h /2) T_H$ by Equation~\eqref{eq:SamplingError-Heavy}.

For the first term, we have $\lvert \sum_i (F_H^{(i)} - D_H^{(i)} \rvert) /(ip) \le  \delta \sum_i D_H^{(i)}/(ip) $ by Equation~\eqref{eq:F0Error-Heavy1}.
Again by Equation~\ref{eq:SamplingError-Heavy}, 
$\sum_i D_H^{(i)} /(ip) \le T_H + (\varepsilon_h/2) T_H = (1 + \varepsilon_h/2)T_H$.
Therefore, $\sum_i \lvert (F_H^{(i)} - D_H^{(i)} )/(ip) \rvert \le (\varepsilon_h/2) (1 + \varepsilon_h/2) T_H $.
Combining the two terms gives
\[
     \lvert \sum_i F_H^{(i)}/(ip) - T_H \rvert \le
      \delta (1 + \varepsilon_h/2) T_H + (\varepsilon_h/2) T_H.
\]
With $\delta = \varepsilon/3$ and $\varepsilon_h \le 1$,
$\delta( 1+ \varepsilon_h/2) + \varepsilon/2 \le 5\varepsilon_h/6 + \varepsilon_h^2/6 \le \varepsilon_h$.
Therefore, $\lvert \sum_i F_H^{(i)}/(ip) - T_H \rvert \le \varepsilon_h T_H$.

Finally, by the  union-bound over the sampling error and the$F_0$ error, we have
\[
    \Pr\left(\lvert \sum_i F_H^{(i)}/(ip) - T_H \rvert  \le \varepsilon_h T_H \right) \ge 98/100.\qedhere
\]
\end{proof}
\begin{corollary}\label{cor:error-2pass}
    Let $q = 3p^2 - 2p^3$ and define $Y_L = F_L/q$ and
    $Y_H = \sum_i F_H^{(i)}/(ip)$.
For a given error parameter $\varepsilon > 0$, 
let $p = \beta \log n/(\varepsilon^2 \sqrt{\numtlb})$ and
$k = \Theta(\log /\varepsilon^2)$, for sufficiently large constant $\beta$. Then, the final estimator 
$\hat T = Y_L + Y_H$ satisfies
\[
    \Pr\left(\lvert \hat T - T \rvert \le \varepsilon T \right) \ge 97/100.
\]
\end{corollary}
\begin{proof}[Proof of Corollary~\ref{cor:error-2pass}]
Let $T_L$ and $T_H$ respectively denote the number of light and heavy triangles. Then $ T = T_L + T_H$.
Apply the light error guarantee (Lemma~\ref{lem:errorProb}) with internal error
$\varepsilon/3$
\begin{equation}\label{eq:light-Traingle-Prb}
    \Pr(\lvert Y_L - T_L \rvert \le (\varepsilon/3) T) \ge 99/100. 
\end{equation}
Apply Lemma 38 with internal error $\varepsilon/3$,
$\varepsilon/3$
\begin{equation}\label{eq:Heavy-Traingle-Prb}
    \Pr(\lvert Y_H - T_H \rvert \le (\varepsilon/3) T_H) \ge 98/100. 
\end{equation}
By the union bound on~\eqref{eq:light-Traingle-Prb} and~\eqref{eq:Heavy-Traingle-Prb}, both events hold simultaneously with probability at least $97/100$.
On this event
\[
    \lvert \hat T - T \rvert = \lvert  Y_L - T_L \rvert
    + \lvert Y_H - T_H \rvert \le (\varepsilon/3) T + (\varepsilon/3) T_H = (2\varepsilon/3) T \le \varepsilon T.\qedhere
\]
\end{proof}
For the space used by the algorithm, observe that we store the sets $V_{oracle}$, $E_{oracle}$ and $E_{sampled}$ in addition to 
the limited independence hash function and four $F_0$ counters.
The limited independence hash function uses $O(\log n)$ bits of space.
By Theorem~\ref{thm:F_0BJKST Clipped}, the light $\cKMV_k$ estimator uses $O(\log n)$ bits of space, while each heavy $\cKMV_k$ estimator stores identities of the form $(e, \tau)$, whose universe has size $O(n^3)$
and therefore also uses $O(\log n)$ bits of space. 
 The overall space complexity is thus dominated by the space used to store the sets $V_{oracle}$, $E_{oracle}$ and $E_{sampled}$. The expected space usage is thus 
\[
     \E[|V_{oracle}| + |E_{oracle}| + |E_{sampled}|] + O(\poly {\log n}) 
      \le p \cdot n + \sum_v p \cdot d_v + p \cdot m \le 4 \cdot p \cdot m.
\]
Hence, the expected space complexity of the algorithm is upper bounded by
$\tilde O(\varepsilon^{-2} m /\sqrt T)$.
Hence, we establish \Cref{Thm:2Pass}.

\TwopassTriangle*


\section{Missing Proofs from Section~\ref{sec:4Pass}}\label{sec:Missing-Proofs-4Pass}

\ForwardTriangleCount*

\begin{proof}[Proof of Lemma~\ref{lem:ForwardTriangleCount}]
  The number of triangles $y_e$ associated with an oriented edge $e = \{u,v\}$    ($\hat d_u \le \hat d_v$) is the same as the number $n_w$ of vertices $w$ such that $(u,v,w)$ is a triangle and $\hat d_w \ge \hat d_v$. Since $w$ is a neighbour of $v$, it follows that 
    \begin{equation}\label{eq:N_w1}
            y_e = n_w \le d_v
    \end{equation}
Note that $d_v$ in~\eqref{eq:N_w1} above is the true degree of $v$.        For every $w$ that completes a triangle assigned to  $\{u,v\}$ we must have $\hat d_w \ge \hat d_v$. It follows that
    \begin{equation}\label{eq:degreeClosingVertex}
            d_w \ge (1 - \varepsilon) \hat d_w \ge (1 - \varepsilon) \hat d_v,   
    \end{equation}
since $\hat d_w \le (1 + \varepsilon) d_w$ implies
$d_w \ge \hat d_w/ (1 + \varepsilon) \ge (1 - \varepsilon) \hat d_w$.
Now we bound the number $n_w$ of vertices $w$ satisfying the inequality~\eqref{eq:degreeClosingVertex}. 
We have that 
\begin{equation}
    n_w \cdot (1 - \varepsilon) \cdot  \hat d_v \le \sum_{w \mid w \text{~satisfies \eqref{eq:degreeClosingVertex}}} d_w \le \sum_{w} d_w = 2m
\end{equation}
Therefore, we get that
\begin{equation}\label{eq:N_w2}
 n_w \le \frac{2m}{(1- \varepsilon) \hat d_v}.
\end{equation}
Combining~\eqref{eq:N_w1} and~\eqref{eq:N_w2} we get that
\begin{align*}
    \begin{aligned}
    y_e &\le  \min \left (d_v, \frac{2m}{(1- \varepsilon) \hat d_v} \right)
         \le  \min \left (d_v, \frac{2m}{(1- \varepsilon) (1- \varepsilon) d_v} \right) 
          = \min \left (d_v, \frac{2m}{(1- \varepsilon)^2 d_v} \right)
    \end{aligned}
\end{align*}
In either case, when $\min \left (d_v, \frac{2m}{(1- \varepsilon)^2 d_v} \right) = d_v$ or when  $\min \left (d_v, \frac{2m}{(1- \varepsilon)^2 d_v} \right) = \frac{2m}{(1- \varepsilon)^2 d_v}$, we have that 
\[
   y_e  \le \min \left (d_v, \frac{2m}{(1- \varepsilon)^2 d_v} \right) \le 
   \sqrt{2m} /1 - \varepsilon = O_{\varepsilon}(\sqrt m).
\]   
\end{proof}

\VarSampling*

 \begin{proof}[Proof of Claim~\ref{Claim: Variance Due to sampling}]
 
 Given a fixed degree estimates vector $\hat d$, the randomness 
    comes from the random sampling of the edge $e = \{u,v\}$ and $r = \lceil \frac{\min \{\hat d_u , \hat d_v \}}{(1- \varepsilon) \sqrt m} \rceil$ uniformly sampled random neighbours of 
     pivot.
     For a fixed $\hat d$ and selected edge $e$, we can write $ Y \mid \hat d, e = \frac{1}{r} \hat d_{pivot} \sum_{k = 1}^{r} A_k $, where $A_k$ is an indicator random variable which is $1$ when a sampled neighbour $w_k$ of pivot forms a triangle with $e$ and 
     $w_k \succ \max\{ u, v\}$ in the order induced by $\hat d$ and vertex IDs.
     It follows that $A_k \sim \textbf{Bernoulli}(p)$, where $ p = p(e) = y_e/d_{pivot}$. We can bound the variance of $Y \mid \hat d , e$ as
\begin{align}
    \begin{aligned}
        \Var(Y \mid \hat d, e) &= \Var\left(\frac{1}{r} \hat d_{pivot} \sum_{k = 1}^{r} A_k \right) = \frac{\hat d^2_{pivot}}{r^2} \sum_{k = 1}^{r} \Var(A_k)\\
                             & = \frac{\hat d^2_{pivot}}{r} p(1 - p)
                             \le \frac{\hat d^2_{pivot}}{r} \cdot \frac{y_e}{d_{pivot}}.
    \end{aligned}
\end{align} 
Since $\esttriangle = m \cdot Y$, it follows that $\Var(\esttriangle \mid \hat d, e) = m^2 \Var(Y \mid \hat d, e)$.
Using $r \ge \frac{\hat d_{pivot}}{(1 - \varepsilon) \sqrt m}$, we get that
\begin{equation}\label{eq:Variance-cond-de}
        \Var(\esttriangle \mid \hat d , e) \le 
     (1 - \varepsilon) m^{5/2} \frac{\hat d_{pivot}}{d_{pivot}} y_e.
\end{equation}
Now, we invoke the law of total variance over the sampled edge
\begin{equation}\label{eq:var-est-de-total}
    \Var(\esttriangle \mid \hat d) = \E_e[\Var(\esttriangle \mid \hat d, e)]  + 
    \Var_e (\E[\esttriangle \mid \hat d, e]).
\end{equation}
Recall that $\mathcal E_{tight}$ is the event 
that the degree estimates of all the vertices of graph are within $(1 \pm \varepsilon)$ of their true values. i.e., on $\mathcal E_{tight}$, we have
$\frac{\hat d _{pivot}}{d_{pivot}} \le 1 + \varepsilon$.
Thus, averaging~\eqref{eq:Variance-cond-de} over the 
$m$ uniform edges,
\begin{align}\label{eq:exp-of-var}
  \begin{aligned}
    \E_e[\Var(\esttriangle \mid \hat d, e)] &\le (1 + \varepsilon) (1- \epsilon) \frac{m^{5/2}}{m} \sum_e y_e  \le (1 + \varepsilon) m^{3/2} T.
    \end{aligned}
\end{align}
Also, on $\mathcal E_{tight}$, Lemma~\ref{lem:ForwardTriangleCount} gives us
$y_e \le C_{\varepsilon} \sqrt m$.
Bounding the second term of~\eqref{eq:var-est-de-total}, for fixed $\hat d$, we have 
$\E[\esttriangle \mid d, e] = m \frac{\hat d_{pivot}}{d_{pivot}} y_e$.
Therefore, on $\mathcal E_{tight}$, 
\begin{align}\label{eq:var-of-expected}
   \begin{aligned}
    \Var(\E[\esttriangle \mid d, e]) \le
         m^2 \E_e \left [\left( \frac{\hat d_{pivot}}{d_{pivot}} y_e \right)^2 \right]
         \le m^2 \frac{1}{m} (1 + \varepsilon)^2 \sum_e y_e^2 
          \le C_{\varepsilon} 
         (1 + \varepsilon)^2 m^{3/2} T.
    \end{aligned}
\end{align}
Hence, on $\mathcal E_{tight}$, following~\eqref{eq:exp-of-var} and~\eqref{eq:var-of-expected}, we have
\begin{equation}\label{eq:variance-good-event}
   \Var(\esttriangle \mid \hat d) \le C_1 m^{3/2} T.
\end{equation}
Next, we analyze what happens when $\mathcal E_{tight}$
does not hold. On $\mathcal E_{tight}^c$, we use clipping and the following holds for any estimate $\hat d_v$
$0 \le \hat d_v \le 2(n-1)$.
Since the random variable $Y$ is an average of $r$  terms, each of which is either $0$ or $\hat d_{pivot}$, we  have 
$0 \le Y \le \hat d_{pivot} \le 2(n-1)$.

Hence
$0 \le \esttriangle = m Y \le 2 m(n-1)$.
Consequently, on $\mathcal E_{tight}^c$
$\Var(\esttriangle \mid \hat d) \le \E[\esttriangle^2 \mid \hat d] \le 4m^2 (n-1)^2$.

We pick $k = O(\log n/\varepsilon^2)$ so that
$\Pr(\mathcal E_{tight}^c) \le \varepsilon n^{-6}$.
Therefore, the contribution of the event $\mathcal E_{tight}$ to the variance bound is given by
\begin{equation}\label{eq:variance-bad-event}
     \E_{\hat d} \left[\Var(\esttriangle \mid \hat d \mathbf 1_{\mathcal E_{tight}^c}) \right] \le 
     4 m^2 (n - 1)^2 \cdot \varepsilon n ^{-6} = O(\varepsilon).
\end{equation}
Since $T \ge 1$, this is $O(\varepsilon T) = O(m^{3/2} T)$.
Combining the contribution of the event $\mathcal E_{tight}^c$ to the variance from Equation~\eqref{eq:variance-bad-event} to the contribution of the complementary good event $\mathcal E_{tight}$ from
Equation~\eqref{eq:variance-good-event}, we get that
\[
     \E_{\hat d} \left [ \Var(\esttriangle \mid \hat d) \right] \le C_1 m^{3/2} T + O(\varepsilon) \le 
     C' m^{3/2} T . \qedhere
\]
\end{proof}

\spacecompfourpass*

\begin{proof}[Proof of Lemma~\ref{lem:spacecomp4pass}]
   A single run of Algorithm~\ref{Algo: Triangle with 4-pass multi-occurrence} stores one sampled edge in the first pass which requires $O(\log n)$ bits.
   In the second pass, the algorithm maintains $\cKMV_{k}$ for the two endpoints of the sampled edge with $k = c\varepsilon^{-2}(\log n + \log \frac{1}{\varepsilon})$. 
    In the third pass it samples $r$ neighbours of the pivot vertex using $r$ $\ell_0$- samplers. Each of the $r$ samplers requires $\widetilde O(\poly \log (n))$ bits. Finally, in the fourth pass it maintains one $\cKMV_k$ estimator for each of these $r$ sampled neighbours and stores one $Z_k$ value of 
    $O(\log n)$ bits for each sampled neighbour.
    
    Below we bound the expected value of $r$
   to bound the space required by $(r +2)$ many $\cKMV_k$ estimators and $r$ many $\ell_0$ samplers.

   For an edge $e$, let $r_e$ denote the value assigned to $r$ 
   by the algorithm, if $e$ is sampled. Then, 
   $\E[r] = \frac{1}{m} \sum_{e \in E} r_e$.
    For the sampled edge $e = \{u, v\}$,
   $r_e \le 1 + \min \{ \hat d_u , \hat d_v \}/(1 - \varepsilon) \sqrt m$, where $+1$ accounts for the ceiling in the definition of $r$. 
   Let $\mathcal E_{tight}$ denote the event that all degree estimates satisfy 
   \[
        (1- \varepsilon) d_v \le \hat d_v \le (1 + \varepsilon) d_v,
   \]
   Since $\hat d_v \le (1 + \varepsilon) d_v$ for every vertex $v$ on $\mathcal E_{tight}$, it follows that
   \begin{equation}\label{eq:r-Condition-Tight}
         \E[r \mid \mathcal E_{tight}] \le 1 +  \frac{1}{m \sqrt
         m}  O\left (\frac{1 + \varepsilon}{1 - \varepsilon}\right) \sum_{e = \{u, v\}} \min \{d_u, d_v \}.
   \end{equation}
   Now, $\min \{d_u, d_v\} \le \sqrt{ d_u d_v}$, and by Cauchy-Schwarz, 
   \begin{equation}\label{eq:sum-of-sqrts}
       \sum_{e = \{ u, v\}}  \sqrt{d_u d_v} \le 
       \sqrt m \sum_e d_u d_v = O(m^{3/2}).
   \end{equation}
Hence, it follows from Equations~\eqref{eq:r-Condition-Tight} and~\eqref{eq:sum-of-sqrts} that
\begin{equation}\label{eq:r-Expected-OnTight1}
     \E[r \mid \mathcal E_{tight}]  \le 1 + O\left(\frac{1 + \varepsilon}{1 - \varepsilon}\right) \cdot O(1) = O\left(\frac{1 + \varepsilon}{1 - \varepsilon}\right).
\end{equation}
For $0 \varepsilon \le 1/2$, $\frac{1 + \varepsilon}{1 - \varepsilon} \le 3$, so
\begin{equation}\label{eq:r-Expected-OnTight}
\E[r \mid \mathcal E_{tight}] = O(1).
\end{equation}

Now, let us bound the expected value of $r$ on the complementary event $\mathcal E_{tight}^c$. In the bad event the degree estimates are clipped at $2(n-1)$. Therefore $r = O(n)$.
By \Cref{Cor: Clipped KMV Union}, we have $\Pr(\mathcal E_{tight}^c) \le \varepsilon n^{-6} \le n^{-6}$. 
Therefore,
\begin{equation}\label{eq:r-Expected-OnNTight}
\E[r \mid \mathcal E_{tight}^c] \le O(n) \cdot n^{-6}.
\end{equation}
It follows by Equation~\eqref{eq:r-Expected-OnTight} and~\eqref{eq:r-Expected-OnNTight} that
\[
    \E[r] = \E[r \mid \mathcal E_{tight}] + \E[r \mid \mathcal E_{tight}^c] \le O(1) + O(n) \cdot n^{-6} = O(1).
\]
Therefore, the expected number of $\cKMV_k$ estimators is 
$\E[r + 2] = O(1)$. Each $\cKMV_k$ estimator uses $O(\varepsilon^{-2}(\log n + \log \frac{1}{\varepsilon}) \log n)$ bits of space. Consequently the expected space for all $r + 2$  $\cKMV_k$ estimators is $\widetilde{O}\left(\varepsilon^{-2}\log \frac{1}{\varepsilon}\right)$.
The $\ell_0$ samplers, sampled edge, and $Z_k$ values contribute only polylogarithmic additional space.
Hence the total expected space usage of Algorithm~\ref{Algo: Triangle with 4-pass multi-occurrence} is 
$\widetilde{O}\left(\varepsilon^{-2}\log \frac{1}{\varepsilon}\right)$
up to polylogarithmic hidden factor contributed by the $\ell_0$ sampler implementation.

\end{proof}


\section{Cliques in Multi-Pass}\label{AppSec: Cliques}

In this section, we present a generalization of \Cref{Algo: Triangle with 4-pass multi-occurrence} to the general case of $k$-cliques with $k = O(1)$. We denote $\numcliques{k}$ to be the number of $k$-cliques present in the graph.

\begin{algorithm}[ht]
    \caption{Four-pass $(2k+1)$-Cliques Counting for Repeated-Edge Arrival}\label{Algo: Odd Cliques with 4-pass multi-occurrence}
    \begin{algorithmic}[1]
        \Require $\vertexcount, \edgecount$
         \State \textbf{Pass 1:}
          \State Select $k$ edges $\sbrac{e_i = (u_i,v_i)}_{i=1}^k$ uniformly at random from the stream
                   using $\ell_0$ samplers.\label{Lin:EdgeSampled - Cliques}
          \If{$\sbrac{u_1,v_1,\ldots,u_k,v_k}$ does not contain $2k$ distinct vertices}
          \State \Return $\estcliques{2k+1} = 0$ \Comment{Algorithm ends here in this case.}
          \EndIf
          \State \textbf{Pass 2:}
            \State Compute $\hat d_{u_i}$ and $\hat d_{v_i}$ using  $\cKMV$ estimators for $d_{u_i}$ and $d_{v_i}$ for $i \in [k]$. \label{lin:pivotDegreesEst - Cliques}
            \State \textbf{Pass 3:}
              \State By renaming vertices if necessary, ensure that $u_i \prec v_i$ for all $i \in [k]$
              \State Let $r \gets \lceil \min\{\hat d_{u_1}, \hat d_{v_1}\} /(1-\varepsilon)\sqrt m \rceil$.\label{lin:mindegreeclq}
              \State $pivot \gets \arg \min \{\hat d_{u_1}, \hat d_{v_1} \} $.
              \For { $l \gets 1 \text{~to~} r$}
                 \Comment{We use an $\ell_0$ sampler to sample uniformly from the set of distinct neighbours of $pivot$.}
                \State Sample a vertex $w_l$ using an $\ell_0$ sampler from $N_{pivot}$.
            \EndFor
             \State \textbf{Pass 4:}
               \State Compute $\hat d_{w_1},\hat d_{w_2}, \ldots, \hat d_{w_r}$ using $\cKMV$ estimators.
               \For {$ \ell \gets 1 \text{~to~} r$}
                  \State  $Z_\ell \gets 0$.
                    \If{$(u_1,v_1,u_2,v_2,\ldots,u_k,v_k,w_\ell) \text{~form a clique and~} w_\ell \succ v_k \succ u_k \succ \ldots \succ v_1\succ u_1 $}\label{lin:CliqueDetected4Pass}
                   \State  $Z_\ell \gets \hat  d_{pivot}$
                   \EndIf
               \EndFor
                \State $Y = \frac{1}{r} \sum_{l =1}^{r} Z_\ell$
                \State $\estcliques{2k+1} = m^k \cdot Y$
                \State \Return $\estcliques{2k+1}$
    \end{algorithmic}
\end{algorithm}


\begin{lemma}\label{Lem:4Pass-Expected Clques}
$\E[\estcliques{2k+1}] = \numcliques{2k+1}\fbrac{1 \pm\frac{\varepsilon}{2}}$.
\end{lemma}
\begin{proof}[Proof of Lemma~\ref{Lem:4Pass-Expected Clques}]
Let $\mathcal E_{e_1,\ldots,e_k}$ be the event that the edges
$\sbrac{e_i = \{ u_i, v_i\}}_{i = 1}^k$ is sampled in Line~\ref{Lin:EdgeSampled - Cliques}.
Every time a clique is detected in line~\ref{lin:CliqueDetected4Pass}, the variable $Z_\ell$ is assigned the estimated degree $\hat d_{pivot}$ of the pivot vertex. Let $A_{\ell}$ be the indicator 
that a clique with vertex $w_\ell$ is detected in line~\ref{lin:CliqueDetected4Pass}.
Then $$\E[A_\ell \mid \mathcal E_{e_1,\ldots,e_k}  ] = \frac{y_{e_1,\ldots,e_k}}{d_{pivot}} $$
where $y_{e_1,\ldots,e_k}$ is the number of cliques associated with edges $e_1,\ldots,e_k$ according to the ordering defined by estimated degrees and $d_{pivot}$ is the true degree of the pivot vertex. We have the true degree in the denominator since the $\ell_0$ sampler will 
sample a neighbour uniformly from the neighbourhood of $pivot$ although we don't know the true value of its size. Then, we have

$$ \E[Z_{\ell} \mid \mathcal E_{e_1,\ldots,e_k}]  = \hat d_{pivot} \cdot \frac{y_{e_1,\ldots,e_k}}{d_{pivot}}$$

We now focus on the expectation of $\hat{d}_{pivot}$. By \Cref{Cor: Clipped KMV Union}, we have for an appropriately sized $\cKMV$ estimators
\begin{align}
    \E[\hat{d}_{pivot}] &= \E[\hat{d}_{pivot}\mid\mathcal{E}_{tight}] + \E[\hat{d}_{pivot}\mid\mathcal{E}_{tight}^c]\nonumber\\
    &\leq \left(1+\frac{\varepsilon}{4}\right)d_{pivot}(1 - \varepsilon n^{-6})+2(n-1)\varepsilon n^{-6} \leq  \left(1 +\frac{\varepsilon}{2}\right)d_{pivot}\label{Eq: Clique clipped degree ub}
\end{align}
Similarly, we have
\begin{align}
    \E[\hat{d}_{pivot}] &= \E[\hat{d}_{pivot}\mid\mathcal{E}_{tight}] + \E[\hat{d}_{pivot}\mid\mathcal{E}_{tight}^c]\nonumber\\
    &\geq \left(1-\frac{\varepsilon}{4}\right)d_{pivot}(1 - \varepsilon n^{-6}) \geq  \left(1 - \frac{\varepsilon}{2}\right)d_{pivot}\label{Eq: Clique clipped degree lb}
\end{align}
Hence, combining \cref{Eq: Clique clipped degree lb,Eq: Clique clipped degree ub}, we have:
\begin{align*}
    \E[\hat{d}_{pivot}] \in \left[\fbrac{1 - \frac{\varepsilon}{2}}d_{pivot}, \fbrac{1 +\frac{\varepsilon}{2}}d_{pivot}\right]
\end{align*}

Then, we have 

$$\E[Z_\ell \mid \mathcal E_{e_1,\ldots,e_k}  ] = \frac{y_{e_1,\ldots,e_k}}{d_{pivot}} \E[\hat d_{pivot}] = y_{e_1,\ldots,e_k}\fbrac{1 \pm\frac{\varepsilon}{2}}$$
It follows that $\E[Y \mid \mathcal E_{\sbrac{e_i}_{i = 1}^k} ] = y_{e_1,\ldots,e_k}\fbrac{1 \pm\frac{\varepsilon}{2}}$.
Then,
\[
\E[\estcliques{2k+1}] = m^k \cdot \E[Y] = m^k \cdot \sum_{e_1,\ldots,e_k} \frac{1}{m^k}\E[Y \mid \mathcal E_{e_1,\ldots,e_k}] = \numcliques{2k+1}\fbrac{1 \pm\frac{\varepsilon}{2}}\qedhere
\]
\end{proof}

\begin{lemma}[Clique Assignment Lemma]\label{Lem: Clique Assignment Lemma}
If the estimated degree $\hat d_v$ for every $v \in V$ is within $(1 \pm \varepsilon)$ multiplicative error from the true degree of $v$, then
\[y_{e_1,\ldots,e_k} \le \frac{\sqrt {2m}}{\sqrt {1 - \varepsilon^2}} = O_\varepsilon(\sqrt m).\]
\end{lemma}

\begin{proof}[Proof of \Cref{Lem: Clique Assignment Lemma}]
  The number of cliques $y_{e_1,\ldots,e_k}$ associated with an oriented edge set $\sbrac{e_i = \{u,v\}}_{i=1}^k$    ($\hat d_{u_i} \le \hat d_{v_i}$) is the same as the number $n_w$ of vertices $w$ such that $(u_1,v_1,\ldots,u_k,v_k,w)$ is a clique and $\hat d_w \ge \hat d_{v_k} \ge \hat d_{u_k} \ldots \ge \hat d_{v_1} \ge \hat{d_{u_1}}$. Since $w$ is a neighbour of $v_1$, it follows that 
    \begin{equation}\label{eq:N_w1Cliques}
            y_{e_1,\ldots,e_k} = n_w \le d_{v_1}
    \end{equation}
Note that $d_{v_1}$ in~\eqref{eq:N_w1Cliques} above is the true degree of $v_1$.        For every $w$ that completes a clique assigned to  $\sbrac{e_i = \{ u_i, v_i\}}_{i = 1}^k$ we must have $\hat d_w \ge \hat d_{v_1}$. It follows that
    \begin{equation}\label{eq:degreeClosingVertexclq}
            d_w \ge (1 - \varepsilon) \hat d_w \ge (1 - \varepsilon) \hat d_{v_1}    
    \end{equation}
Now we bound the number $n_w$ of vertices $w$ satisfying the inequality~\eqref{eq:degreeClosingVertexclq}. 
We have that 
\begin{equation*}
    n_w \cdot (1 - \varepsilon) \cdot  \hat d_{v_1} \le \sum_{w \mid w \text{~satisfies \eqref{eq:degreeClosingVertexclq}}} d_w \le \sum_{w} d_w = 2m
\end{equation*}
Therefore, we get that
\begin{equation}\label{eq:N_w2clq}
 n_w \le \frac{2m}{(1- \varepsilon) \hat d_{v_1}}
\end{equation}
Combining~\eqref{eq:N_w1Cliques} and~\eqref{eq:N_w2clq} we get that
\begin{align*}
    \begin{aligned}
    y_{e_1,\ldots,e_k} &\le  \min \left (d_{v_1}, \frac{2m}{(1- \varepsilon) \hat d_{v_1}} \right)
         \le  \min \left (d_{v_1}, \frac{2m}{(1- \varepsilon) (1+ \varepsilon) d_{v_1}} \right) 
          = \min \left (d_{v_1}, \frac{2m}{(1- \varepsilon^2) d_{v_1}} \right)
    \end{aligned}
\end{align*}
Then, we have that 
\[
   y_{e_1,\ldots,e_k}  = \min \left (d_{v_1}, \frac{2m}{(1- \varepsilon^2) d_{v_1}} \right) \le 
   \sqrt{2m} /\sqrt{1 - \varepsilon^2}.\qedhere
\]   
\end{proof}

\begin{lemma}\label{Lem:4PassVarianceClqiues}
$\Var(\estcliques{2k+1}) = O(m^{(k+1)/2} \cdot \numcliques{2k+1})$.
\end{lemma}

\begin{proof}[Proof of Lemma~\ref{Lem:4PassVarianceClqiues}]
    The variance in the value of $\estcliques{2k+1}$ is caused by two sources. 
\begin{itemize}
    \item The first is the variance due to the randomness of the sampling process, i.e., in the edge selection and the neighbourhood sampling of the pivot vertex. Observe that these are also the sources of variance when we use exact degrees instead of estimated degrees and uniform sampling without $\ell_0$-samplers.

    \item The second is the variance caused by degree estimators and $\ell_0$-samplers.
\end{itemize}
Let $\hat d$ be a vector of all the degree estimates of a given run of the algorithm.
Using the law of total variance, we can decompose $\Var(\estcliques{2k+1})$ as:
\begin{equation}\label{eq:TotalVaraince4PassClqs}
      \Var(\estcliques{2k+1}) = \E_{\hat d}\left[\Var\left(\estcliques{2k+1} \mid \hat d \right)\right] + \Var_{\hat d} \left(\E[\estcliques{2k+1} \mid \hat d]\right).
\end{equation}
The first term is the expected variance of $\estcliques{2k+1}$ conditioned on a fixed vector of all the degree estimates in a single run of the algorithm.
The second term accounts for the variance in the expected estimate $\E[\estcliques{2k+1} \mid \hat d]$ caused by the degree estimates.
We will bound these two separately. 

 First, we bound $\E_{\hat d}\left[\Var\left(\estcliques{2k+1} \mid \hat d\right )\right]$. We prove the following claim
 \begin{claim}\label{Claim: Clique Variance Due to sampling}
 There exists a constant $C'$ such that
 \[
     \E_{\hat d}\left[\Var\left(\estcliques{2k+1} \mid \hat d\right )\right] \le  C' \cdot m^{k+1/2} \numcliques{2k+1}.
 \]
 \end{claim}

 \begin{proof}[Proof of Claim~\ref{Claim: Clique Variance Due to sampling}]
 
 Given a fixed vector $\hat d$, the randomness is 
    coming from the selection of the edges $\sbrac{e_i}_{i = 1}^k$ and $r = \ceil{\frac{\hat d_{pivot}}{(1- \varepsilon) \sqrt m}}$ random neighbours of 
     pivot vertex.
     For a fixed $\hat d$ and selected edge set $\sbrac{e_i}_{i = 1}^k$, we can write $ Y \mid \hat d, \sbrac{e_i}_{i = 1}^k = \frac{1}{r} \hat d_{pivot} \sum_{k = 1}^{r} A_\ell $, where $A_\ell$ is an indicator random variable which is $1$ when a sampled neighbour $w_\ell$ of pivot forms a clique with $\sbrac{e_i}_{i = 1}^k$ and 
     $w_\ell \succ \cup_{i=1}^k \sbrac{u_i,v_i}$ in the order induced by $\hat d$ and vertex IDs.
     It follows that $A_\ell \sim \mathrm{Bernoulli}(p)$, where $ p = p({e_1,\ldots,e_k}) = y_{e_1,\ldots,e_k}/d_{pivot}$. We can bound the variance of $Y \mid \hat d , \sbrac{e_i}_{i = 1}^k$ as
\begin{align}
    \begin{aligned}
        \Var(Y \mid \hat d, \sbrac{e_i}_{i = 1}^k) &= \Var\left(\frac{1}{r} \hat d_{pivot} \sum_{\ell = 1}^{r} A_\ell \right) = \frac{\hat d^2_{pivot}}{r^2} \sum_{\ell = 1}^{r} \Var(A_\ell)\\
                             & = \frac{\hat d^2_{pivot}}{r} p(1 - p)
                             \le \frac{\hat d^2_{pivot}}{r} \cdot \frac{y_{e_1,\ldots,e_k}}{d_{pivot}} 
    \end{aligned}
\end{align} 
Since $\estcliques{2k+1} = m^k \cdot Y$, it follows that
\[\Var(\estcliques{2k+1} \mid \hat d, \sbrac{e_i}_{i = 1}^k) = m^{2k} \Var(Y \mid \hat d, \sbrac{e_i}_{i = 1}^k).\] 

Now, we bound the variance of $\estcliques{2k+1}$ without the conditioning on edge set $\sbrac{e_i}_{i = 1}^k$. We use the law of total variance again.
\begin{align}\label{Eq: Varaince Total 2 Cliques}
    \begin{aligned}
          \Var(\estcliques{2k+1} \mid \hat d) &= \E_e[\Var(\estcliques{2k+1} \mid \hat d, \sbrac{e_i}_{i = 1}^k)] + \Var_e(\E[\estcliques{2k+1} \mid \hat d, \sbrac{e_i}_{i = 1}^k])          
    \end{aligned}
\end{align}
Let us first bound the first term of Equation~\eqref{Eq: Varaince Total 2 Cliques}.
\begin{align*}
\begin{aligned}
        \E_e[\Var(\estcliques{2k+1} \mid \hat d, \sbrac{e_i}_{i = 1}^k)] &= \frac{1}{m^{k}} \sum_{{e_1,\ldots,e_k}} \Var(\estcliques{2k+1} \mid \hat d,\sbrac{e_i}_{i = 1}^k))\\
        &\le \frac{1}{m^k} \sum_{{e_1,\ldots,e_k}} m^{2k} \cdot \frac{\hat d^2_{pivot}}{r} \cdot \frac{y_{e_1,\ldots,e_k}}{d_{pivot}}\\
        &\le m^k \sum_{{e_1,\ldots,e_k}} \sqrt m \hat d_{pivot} \cdot \frac{y_{e_1,\ldots,e_k}}{d_{pivot}}\\
        &= m^{k+1/2} \sum_{{e_1,\ldots,e_k}} \frac{\hat d_{pivot}}{d_{pivot}} y_{e_1,\ldots,e_k}\\
        &= m^{k+1/2} \numcliques{2k+1} \frac{\hat d_{pivot}}{d_{pivot}}
\end{aligned}    
\end{align*}
The second inequality follows from $r \ge \frac{\hat d_{pivot}}{\sqrt m}$.
Now, we upper bound $\E[\hat d_{pivot}]$
\begin{align*}
    \E[\hat{d}_{pivot}] &= \E[\hat{d}_{pivot}\mid\mathcal{E}_{tight}] + \E[\hat{d}_{pivot}\mid\mathcal{E}_{tight}^c]\\
    &\leq (1+\frac{\varepsilon}{4})d_{pivot}(1 - \varepsilon n^{-6})+2(n-1)\varepsilon n^{-6} \leq  (1 +\frac{\varepsilon}{2})d_{pivot}
\end{align*}
Thus it follows that
\begin{align}
        \E_d[\E_e[\Var(\estcliques{2k+1} \mid \hat d, \sbrac{e_i}_{i = 1}^k)]] \le (1 + \frac{\varepsilon}{2}) \cdot  m^{3/2} \numcliques{2k+1}.\label{Eq: ExpectedVarainceEstimate2 Cliques}
\end{align}
Now, we bound the second term of~\eqref{Eq: Varaince Total 2 Cliques}. First consider $\E[\estcliques{2k+1} \mid \hat d, \sbrac{e_i}_{i = 1}^k]$.
\begin{align*}
        \E[\estcliques{2k+1} \mid \hat d, \sbrac{e_i}_{i = 1}^k] = m^k \cdot \hat d_{pivot} \frac{\sum_{k = 1}^r \E[A_\ell]}{r} = m^k \cdot \hat d_{pivot} \cdot \frac{y_{e_1,\ldots,e_k}}{d_{pivot}}       
\end{align*}
Now, let us bound the variance of $\E[\estcliques{2k+1} \mid \hat d, e]$ over the random selection of edge $e$.
\begin{align}
        \Var_{e_1,\ldots,e_k}(\E[\estcliques{2k+1} \mid \hat d, e]) &= \Var_{e_1,\ldots,e_k} \left(  m^k \cdot \hat d_{pivot} \cdot \frac{y_{e_1,\ldots,e_k}}{d_{pivot}}  \right )\nonumber\\
        &= m^{2k} \Var_{e_1,\ldots,e_k} \left ( \hat d_{pivot} \cdot \frac{y_{e_1,\ldots,e_k}}{d_{pivot}}\right)\nonumber\\
        &\le m^{2k} \E_{e_1,\ldots,e_k}\tbrac{\left(\hat d_{pivot} \cdot \frac{y_{e_1,\ldots,e_k}}{d_{pivot}} \right)^2 }\nonumber\\
        &\le m^{2k} \frac{\hat d_{pivot}^2}{d_{pivot}^2} \E_{e_1,\ldots,e_k}\tbrac{y_{e_1,\ldots,e_k}^2 }\label{Eq:VarainceConditionalEstimate1 Cliques}
\end{align}


Next, we use the combinatorial \Cref{lem:ForwardTriangleCount}, $\max_e y_{e_1,\ldots,e_k} \le c \cdot \sqrt m $ for a suitable constant $c$. We get that 
\begin{align}\label{Eq:sum of y_e^2 Cliques}
        \sum_e y_{e_1,\ldots,e_k}^2 &\le \max_e y_{e_1,\ldots,e_k} \sum_e y_{e_1,\ldots,e_k} 
                      \le  c \sqrt m \cdot \numcliques{2k+1}.
\end{align}
Plugging~\eqref{Eq:sum of y_e^2 Cliques} into~\eqref{Eq:VarainceConditionalEstimate1 Cliques} we get
\begin{align}
     \Var_e(\E[\estcliques{2k+1} \mid \hat d, e]) \le m^k \cdot \frac{\hat d_{pivot}^2}{d_{pivot}^2} \cdot  c \sqrt m \cdot \numcliques{2k+1} = c \cdot  m^{k + 1/2} \cdot \numcliques{2k+1} \cdot \frac{\hat d_{pivot}^2}{d_{pivot}^2}.\nonumber
\end{align}

Then, we have by \Cref{Cor: Clipped KMV Union}

\begin{align}
    \E_{\hat{d}} \tbrac{\Var_e(\E[\estcliques{2k+1} \mid \hat d, e])} \leq c' \cdot  m^{k + 1/2} \cdot \numcliques{2k+1}\label{Eq:VarainceConditionalEstimate2 Cliques}
\end{align}
Plugging the upper-bounds for the two summands of Equation~\eqref{Eq: Varaince Total 2 Cliques} from ~\eqref{Eq: ExpectedVarainceEstimate2 Cliques} and~\eqref{Eq:VarainceConditionalEstimate2 Cliques} we get that
\begin{equation*}
      \E_{\hat{d}}\tbrac{\Var(\estcliques{2k+1} \mid \hat d)} \le c \cdot m^{k + 1/2} \numcliques{2k+1}.\qedhere
\end{equation*}
\end{proof}

Next, we bound the second term of~\eqref{eq:TotalVaraince4Pass}, 
$\Var_{\hat d} \left(\E[\estcliques{2k+1} \mid \hat d]\right)$.
Here, we have to bound the variance of conditional expectation of output (of one run) due to the randomness in the \cKMV{} degree estimators. We prove the following claim:
\begin{claim}\label{claim: Clique Variance Due to Sketches}
There exists a constant $C''$ such that 
\[
    \Var_{\hat d}\left (\E[\estcliques{2k+1} \mid \hat d]\right) \le
            C'' \cdot m^{k+1/2} \cdot T.
\]
\end{claim}
\begin{proof}[Proof of Claim~\ref{claim: Clique Variance Due to Sketches}]

\begin{align}\label{Eq:expectedEstimateGivenApprxDegrees Cliques}
    \E[\estcliques{2k+1} \mid \hat d] = m^k \E_{e_1,\ldots,e_k}[Y \mid \hat d] = 
     m^k \frac{1}{m^k} \cdot  \sum_{e_1,\ldots,e_k} y_{e_1,\ldots,e_k} \cdot \frac{\hat d_{pivot}}{d_{pivot}} 
\end{align}

Let us denote by $\delta_v$ the relative error  $\frac{\hat d_{v} - d_{v}}{d_{v}}$ of the degree estimate of vertex $v$. We can then write $\hat d_{v}/d_{v} = 1 + \delta_{v} $. We have that $\abs{\E[\delta_{v}]} \leq \varepsilon n$ for sufficiently large $\cKMV{}$ estimators, by Theorem~\ref{thm:F_0BJKST Clipped}). 
Moreover, since the variance of $\hat d_{v} = O_{\varepsilon, \delta}((d_{v})^2)$ (see Theorem~\ref{thm:F_0BJKST Clipped}), it follows that $\Var(\delta_{v}) = O(1)$. Let $\beta$ be a large enough constant such that for any $v \in V$, $\Var(\delta_v) \le \beta$.
Then,
\[
     \E[\estcliques{2k+1} \mid \hat d] = \numcliques{2k+1} + \sum_{e_1,\ldots,e_k} y_{e_1,\ldots,e_k} \cdot \delta_{pivot.}  
\]
It follows that
\begin{equation}\label{Eq:varainceSketchMain Cliques}
     \Var_{\hat d}\left (\E[\estcliques{2k+1} \mid \hat d]\right) = 
           \Var_{\hat d}\fbrac{\sum_{e_1,\ldots,e_k} y_{e_1,\ldots,e_k} \cdot \delta_{pivot}}
\end{equation}
Thus, the variance due to \cKMV{} estimators is a weighted sum of 
 errors where weights are given by $y_{e_1,\ldots,e_k}$.
 By ensuring that the \cKMV{} estimators for different vertices use independent randomness (this does not violate the order, we still have an order defined on vertices by their estimated degrees even if degrees are estimated using different estimators), we can group edges by their pivot vertex, We can define for each vertex $v \in V$, $y_v$ the number of cliques associated with 
 $v$ by the degree estimates
\[
    y_v = \sum_{{e_1,\ldots,e_k} \mid v \text{~is the pivot of ~}e_1 } y_{e_1,\ldots,e_k}
\]
We can now express $\sum_{e_1,\ldots,e_k} y_{e_1,\ldots,e_k} \cdot \delta_{pivot }$ as
\begin{equation}\label{Eq:VarainceSketches2 Cliques}
     \sum_{e_1,\ldots,e_k} y_{e_1,\ldots,e_k} \cdot \delta_{pivot } = \sum_v \delta_v \sum_{{e_1,\ldots,e_k} \mid v \text{~is the pivot of ~}e_1 } y_{e_1,\ldots,e_k} = \sum_v \delta_v \cdot y_v.
\end{equation}
Since for each $v \in V$, $\abs{\E[\delta_v]} \leq \varepsilon n $, $\Var(\delta_v) \le \beta$ and $\delta_v$ are independent, it follows that
\[
     \Var \left (\sum_{e_1,\ldots,e_k} y_{e_1,\ldots,e_k} \cdot \delta_{pivot} \right) = \sum_v y_v^2 \Var (\delta_v) \le \beta \sum_v y_v^2.
\]
 So finally, we have to bound $\sum_v y_v^2$.
We can define $y_{max} = \max_v y_v$, we can write
 \[
     \sum_v y_v^2 \le y_{max} \sum_v y_v = y_{max} \cdot \numcliques{2k+1}
 \]
 Therefore, we get that 
 \[
      \Var_{\hat d}\left (\E[\estcliques{2k+1} \mid \hat d]\right) \le 
      \beta \cdot y_{max} \cdot \numcliques{2k+1}
 \]
To bound $y_{max}$, we note that
\[ y_v = \sum_{e \mid v \text{ is the pivot of ~}e_1} y_{e_1,\ldots,e_k} \le 
              \sum_{{e_1,\ldots,e_k} \text{~ is incident on v}} y_{e_1,\ldots,e_k} \le \Delta \cdot C \sqrt m
\]
where $\Delta = \max_{v \in V} d_v^k$.
The last inequality above follows by combinatorial Lemma 3.4 of~\cite{BeraC/stacs/2017/CliqueCountingEA},
It follows that
\[
     \Var_{\hat d}\left (\E[\estcliques{2k+1} \mid \hat d]\right) \le 
        \beta \cdot C \cdot \sqrt m \cdot  \Delta \cdot  \numcliques{2k+1}.
\]
Now, we can upper-bound $\Delta$ in the above inequality by $m^k$
and we get a variance bound for this part which is of the same order asymptotically as Claim~\ref{Claim: Variance Due to sampling}.
More precisely, 
\begin{equation*}
        \Var_{\hat d}\left (\E[\estcliques{2k+1} \mid \hat d]\right) \le
          \beta \cdot C \cdot \sqrt m \cdot m^k \cdot \numcliques{2k+1} = C' \cdot \beta \cdot m^{k + 1/2} \cdot \numcliques{2k+1}.\qedhere
\end{equation*}
\end{proof}
Thus, the lemma follows by Claim~\ref{Claim: Variance Due to sampling} and Claim~\ref{claim:Variance Due to Sketches}.
\end{proof}

Next, we analyse the space complexity of \Cref{Algo: Odd Cliques with 4-pass multi-occurrence}.

\begin{lemma}\label{lem: expected space}
    The expected space required by a single run of Algorithm~\ref{Algo: Odd Cliques with 4-pass multi-occurrence} is $\tilde{O}\left(\frac{k}{\varepsilon^2} \cdot \log (\frac{1}{\delta})\right)$.
\end{lemma}
\begin{proof}[Proof of Lemma~\ref{lem: expected space}]
   A single run of Algorithm~\ref{Algo: Odd Cliques with 4-pass multi-occurrence} stores $O(k)$ sampled edges in the first pass which requires $O(k\log n)$ bits.
   In the second pass the algorithm estimates the degrees of the endpoints of the sampled edges and in the third pass it samples $r$ neighbours of the pivot vertex using $\ell_0$- samplers. Each of the $r$ samplers requires $ O(\poly {\log n})$ bits. Finally, in the fourth pass it estimates the degrees of the $r$ neighbours sampled in previous pass using \cKMV{} estimators of Theorem~\ref{thm:F_0BJKST Clipped} and stores one random variable $Z_\ell$ of $O(\log n)$
   bits for each of the $r$ neighbours. Below we bound the expected value of $r$
   to bound the space required by $r +O(k)$ $F_0$ \cKMV{} estimators and $r$ $\ell_0$s samplers.

   Recall that the edges $e_i = \{u_i,v_i\}_{i =1}^k$ is selected uniformly at random from the stream. Let $\mathcal E_1$ be the event that all of the $r + O(1)$ \cKMV{} estimators
   successfully estimate the respective degrees within the promised $(1 \pm \varepsilon)$ multiplicative error bound and $\mathcal E_2$ be the event that all of the $\ell_0$- samplers successfully provide a uniform sample.
   It follows that 
   \[
      \E[r \mid \mathcal E_1 \cap \mathcal E_2] = \frac{1}{m} \sum_{\{ u, v\} \in E} \frac{\min \{\hat d_u, \hat d_v\}}{\sqrt m} = (1 + \varepsilon) \cdot O(1) = O(1).
   \]
For a suitable constant $c$, let $\delta = n^{-c}$, then by a union bound
over all the degree estimators and $\ell_0$ samplers.
$\Pr[\mathcal E_1 \cap \mathcal E_2] \ge 1 - n^{c - 2}$.
Note that $r = \frac{\min\{\hat d_u, \hat d_v \} }{\sqrt m} \le n$.
Let $c' = c-2$.
Hence, 
\begin{equation*}
    \begin{aligned}
        \E[r] &= \E[r \mid \mathcal E_1 \cap \mathcal E_2] \cdot \left( 1 - \frac{1}{n^{c'}} \right) + 
                  \E[r \mid (\mathcal E_1 \cap \mathcal E_2)^c] \cdot  1/n^{c'}
              \le O(k) \cdot \left( 1 - \frac{1}{n^{c'}} \right) + O(n) \cdot \frac{1}{n^{c'}}               
    \end{aligned}
\end{equation*}
For $c'$ large enough, we get that 
\begin{equation}\label{eq:space4Passclq}
    \E[r] = O(k).
\end{equation}
Combining Equation~\eqref{eq:space4Passclq} with Theorem~\ref{thm:F_0BJKST Clipped} and Theorem~\ref{thm:Jowhari2011} it follows that
the expected space usage of a single run of Algorithm~\ref{Algo: Odd Cliques with 4-pass multi-occurrence} is $\tilde{O}\left(\frac{k}{\varepsilon^2} \cdot \log (\frac{1}{\delta})\right)$.
\end{proof}

Next, combining \Cref{lem:MoM} with the expected space complexity of a single run and Lemmas~\ref{Lem:4Pass-Expected Clques} and~\ref{Lem:4PassVarianceClqiues} for the expectation and variance of $\estcliques{2k+1}$ we get the following theorem:

\begin{theorem}\label{Thm:4PassMainClqs}
    There exists a $4$-pass algorithm that given an $n$-vertex $m$-edge graph presented as a graph stream that allows duplicates outputs an $(\varepsilon, \delta)$-estimate $\estcliques{2k+1}$ of the number of $2k+1$-cliques
    using space $\widetilde O\left((m^{k+1/2}/T) \frac{\log 1/\delta}{\varepsilon^4} \right)$.

\end{theorem}


We now extend the two-pass algorithm of Bera and Chakrabarti ~\cite{BeraC/stacs/2017/CliqueCountingEA} for counting even cliques to the multiple occurrence setting. Lets $(2k-1)!!$ denote $(2k-1)(2k-3)\ldots1$.

\begin{algorithm}[ht]
    \caption{Two-pass $2k$-Cliques Counting for Repeated-Edge Arrival}\label{Algo: Even Cliques with 4-pass multi-occurrence}
    \begin{algorithmic}[1]
        \Require $\vertexcount, \edgecount$
         \State \textbf{Pass 1:}
          \State Select $k$ edges $\sbrac{e_i = (u_i,v_i)}_{i=1}^k$ uniformly at random from the stream
                   using $\ell_0$ samplers.\label{Lin:EdgeSampled - Even Cliques}
          \If{$\sbrac{u_1,v_1,\ldots,u_k,v_k}$ does not contain $2k$ distinct vertices}
          \State \Return $\estcliques{2k} = 0$ \Comment{Algorithm ends here in this case.}
          \EndIf
          
        \State \textbf{Pass 2:}
                \State $Y \gets 1$ if $\sbrac{u_1,v_1,\ldots,u_k,v_k}$ forms a clique, and $0$ otherwise
                \State $\estcliques{2k} = \frac{m^k \cdot Y}{(2k-1)!!}$
                \State \Return $\estcliques{2k}$
    \end{algorithmic}
\end{algorithm}

\begin{lemma}\label{Lemma: Even Cliques}
    $\E[\estcliques{2k}] = \numcliques{2k}$.
\end{lemma}

\begin{proof}[Proof of Lemma~\ref{Lemma: Even Cliques}]
    Let us denote $Y_i$ to be the indicator random variable denoting whether the $i$-th clique is detected by the algorithm. A clique is detected iff $\sbrac{e_i = (u_i,v_i)}_{i=1}^k$ forms an edge cover for the clique. For a $2k$-clique, there are $(2k-1)!!$ valid edge covers for each $k$-clique. Hence, we have:
    \begin{align*}
        \E[Y_i] = \frac{(2k-1)!!}{m^{k}}
    \end{align*}
    This is due to the fact that the $\ell_0$-samplers are exact. Hence, by summing over all $\numcliques{2k}$ $2k$-cliques, we have:
    \begin{align*}
        \E[Y] = \sum_{i =1}^\numcliques{2k} Y_i = \frac{(2k-1)!!\numcliques{2k}}{m^k}
    \end{align*}
    Hence, we have:
    \[ \E[\estcliques{2k}] = \frac{m^k \cdot \E[Y]}{(2k-1)!!} = \numcliques{2k}
    \qedhere \]
\end{proof}

\begin{lemma}\label{Lemma: Even Cliques Variance}
    $\Var[\estcliques{2k}] = \bigo{m^k\numcliques{2k}}$.
\end{lemma}

\begin{proof}[Proof of Lemma~\ref{Lemma: Even Cliques Variance}]
    By the algorithm's estimate of $\estcliques{2k}$, we have
    \begin{align}
        \Var[\estcliques{2k}] = \frac{m^{2k}}{(2k-1)!!^2}\Var[Y] \leq \frac{m^{2k}}{(2k-1)!!^2}\E[Y^2] \label{Eq: Even Clique 1}
    \end{align}
    Now, we focus on $\E[Y^2]$.
    \begin{align}
        \E[Y^2] = \sum_{i=1}^\numcliques{2k} \E[Y_i^2] + \sum_{i \neq j \in [\numcliques{2k}]}\E[Y_iY_j] = \frac{(2k-1)!!\numcliques{2k}}{m^k} 
    \end{align}\label{Eq: Even Clique 2}
   Here, the last part follows from the fact that $Y_i$ and $Y_j$ can not both be one as no two $2k$-cliques share a $k$ sized edge cover. Thus, by combining \eqref{Eq: Even Clique 1} and \eqref{Eq: Even Clique 2}, we have:
   \begin{align*}
       \Var[\estcliques{2k}] \leq \frac{m^k\numcliques{2k}}{(2k-1)!!}
   \end{align*}
\end{proof}


Now, by combining \Cref{lem:MoM,Lemma: Even Cliques,Lemma: Even Cliques Variance} and the fact that each run of \Cref{Algo: Even Cliques with 4-pass multi-occurrence} uses $\bigot{1}$ space, we obtain the following theorem:

\begin{theorem}\label{Thm:2PassEvenClqs}
    There exists a $2$-pass algorithm that given an $n$-vertex $m$-edge graph presented as a graph stream that allows duplicates outputs an $(\varepsilon, \delta)$-estimate $\estcliques{2k}$ of the number of $2k$-cliques
    using space $\widetilde O\left((m^{k}/T) \frac{\log 1/\delta}{\varepsilon^2} \right)$.
\end{theorem}

Combining \Cref{Thm:2PassEvenClqs,Thm:4PassMainClqs}, we obtain the following result for counting cliques in multi-pass multiple occurrence streaming model.

\begin{theorem}\label{Thm:MultiPassCliques}
    There exists a constant pass algorithm that given an $n$-vertex $m$-edge graph presented as a graph stream that allows duplicates outputs an $(\varepsilon, \delta)$-estimate $\estcliques{k}$ of the number of $k$-cliques
    using space $\widetilde O\left((m^{\nicefrac{k}{2}}/T) \frac{\log 1/\delta}{\varepsilon^4} \right)$.
\end{theorem}
\section{Clipped Estimator Analysis}\label{sec:ClippedEstimator}
We start by presenting the theorem summarizing the bottom $k$ sketch our clipped estimator is based on. We use the bottom-$k$ distinct elements estimator of \cite{cohen1997size,cohen2008tighter}. The following theorem gives the properties of the estimator.

\begin{theorem}[$\KMV$ / bottom-$k$ distinct-count estimator Properties]
\label{thm:F_0BJKST}
Let $ \mathcal U$ be a universe such that $|\mathcal U| = U$. 
    Let a stream over $\mathcal U$ contain $F_0$ distinct elements, where $0 \le F_0 \le U$.
    Hash each distinct item independently to a uniform value in $(0,1)$.
Let $R_{(1)}\le\cdots\le R_{(F_0)}$ denote the order statistics of these hash values, and fix an integer $k\ge 3$.

Define the $\KMV$ estimator
$\widehat{F_0}\;=\;\frac{k-1}{R_{(k)}}$.
Under the above hashing model the estimator satisfies:
\begin{enumerate}
  \item \textbf{Unbiasedness:} \(\displaystyle \mathbb{E}[\widehat{F_0}] = F_0.\)
  \item \textbf{Exact variance:} \(\displaystyle 
    \Var(\widehat{F_0})
    = \frac{F_0\,\big(F_0-k+1\big)}{k-2}.
    \) 
    In particular, for $F_0\gg k$ one has the asymptotic form $\widehat{F}_{0} \approx F_0^2/(k-2)$.
  \item \textbf{Exponential (sub-Gaussian style) tails:} For any $\varepsilon\in(0,1)$, there exists a constant $c > 0$ such that,
  $\Pr \big(|\widehat{F_0}-F_0|\ge \varepsilon F_0\big)
    \le  2\exp \big(-c\,k\varepsilon^2\big)$,
  
\item The \textbf{space complexity:} of the estimator is $O(k \cdot \log  U)$.
\end{enumerate}
  
   Consequently, to guarantee
  \(\Pr(|\widehat{F_0}-F_0|\ge \varepsilon F_0)\le\delta\) it suffices to choose
  $k \ge  \frac{1}{c \varepsilon^2}\log \frac{2}{\delta}
    =  O \big(\frac{1}{\varepsilon^2}\log\frac{1}{\delta}\big)$.
\end{theorem}


Using the $\KMV$-estimator from Theorem~\ref{thm:F_0BJKST} we define the $\cKMV$ estimator that runs the $\KMV$-estimator and outputs the $\min \{F_{output}, C\}$, where $F_{output}$ is the output produced by the $\KMV$-estimator. The following theorem imply \Cref{thm:F_0BJKST Clipped}.

\begin{theorem}[$\cKMV$/bottom-$k$ distinct-count estimator]\label{thm:F_0BJKST Clipped Constructive}
Let $F_0$ be the true number of distinct elements in a stream over a universe $\mathcal U$ of size $U$. 
Let $F_{output}$ be the output of $\KMV$-estimator of~\cite{cohen2008tighter}. For $C = 2U$, 
define our clipped estimator $F_{clipped}$ as 
$F_{clipped} = \min \{F_{output}, C\}$.
For $k > 2$ (where $k$ is the number of sketches used by the $\KMV$-estimator of Theorem~\ref{thm:F_0BJKST}), for $k\ge 3$, the clipped estimator satisfies, 
\begin{enumerate}
    \item \textbf{Expectation}
            $
                  F_0\left( 1- 2e^{-ck/2}) \right) \le \E[F_{clipped}] \le F_0.
            $
    \item \textbf{Variance}
    $
         \Var(F_{clipped}) \le \frac{F_0(F_0 - k + 1)}{k-2}.
    $
    \item \textbf{Tail Bound}
    $
         \Pr \left(\lvert F_{clipped} - F_0 \rvert  \ge \varepsilon F_0 \right) \le 2e^{-c k \varepsilon^2}.
    $
    \item \textbf{Space Complexity} of $F_{clipped}$ is $O(k\log U)$
\end{enumerate}
\end{theorem}

\subsection{Proof of Theorem~\ref{thm:F_0BJKST Clipped Constructive}}

The following lemmas establish the properties of the clipped $F_0$ estimator that we use in our triangle counting analysis.

\begin{lemma}[Expectation of Clipped Estimator]\label{lem:clipped-expectation}
There exists a constant $c' > 0$ such that
\[
      \E[F_{clipped}] = F_0 - O(F_0 e^{-c'k}).
\]
More precisely, for $k \ge 3$,
$F_0 (1 - 2e^{-ck/2}) \le \E[F_{clipped}] \le F_0$.
\end{lemma}
\begin{proof}[Proof of Lemma~\ref{lem:clipped-expectation}]
Let $L$ denote the loss incurred by the clipped estimator due to 
$\KMV$ outputting a very large number away from the true value $d$.
More formally define
\[
      L = \begin{cases}
                 0, \qquad \text{if } F_{output} \le C\\
                F_{output} - C, \quad \text{if } F_{output} > C.
          \end{cases}
\]
Then, we can write $F_{clipped}$ as
\begin{equation}\label{Eq:dClipped-Deinfe}
       F_{clipped} = F_{output} - L.
\end{equation}
Now, taking expectation on both side of Equation~\eqref{Eq:dClipped-Deinfe}, we get
\begin{equation}\label{eq:clipped-EstimatorLoss}
      \E[F_{clipped}] = \E[F_{output}] - \E[L].
\end{equation}
Recall from Theorem~\ref{thm:F_0BJKST} that $\KMV$-estimator is unbiased. Substituting $\E[F_{output}]  = F_0$ in Equation~\eqref{eq:clipped-EstimatorLoss}, we get that
\begin{equation}\label{eq:clipped-actual-loss}
     \E[F_{clipped}] = F_0 - \E[L].
\end{equation}
Next, we bound the expected loss due to clipping.
If $F_{output} \le C$, then, $L = 0$. Only, when 
$F_{output} > C$, we have a loss $L = F_{output} - C \le F_{output}$,
since $C > 0$ and $F_{output} > C$.
So, whenever we incur a loss due to clipping, it is upper bounded by the value $F_{output}$, output by the estimator.

The true count $F_0$ is always at most $U$, i.e.,
$F_0 \le U = C/2$, or $C \ge 2F_0$.
Now, let us say the output got clipped, this can only happen when
$F_{output} > C$ and since $C \ge 2F_0$, we have $F_{output} > 2F_0$.
Thus, clipping can only happen when the $\KMV$-estimator overestimates by a factor of at least $2$.
We can write the random variable $L$ as
\begin{equation}\label{eq:Loss-Indicator}
        L \le F_{output} \mathbf 1_{\{F_{output} > 2F_0 \}},
\end{equation}
where $\mathbf 1_{\{F_{output} > 2F_0 \}}$ is the indicator random variable for the event
$F_{output} > 2F_0 $.
Now, taking expectation of $L$, we get that
\begin{equation}\label{eq:Loss-Indicator-exp}
       \E[L] \le \E[ F_{output} \mathbf 1_{\{F_{output} > 2F_0 \}}].
\end{equation}
Now, applying Caucy-Shwarz ($\E[XY] \le \sqrt{\E[X^2] \E[Y^2]}$) to the right-hand-side of equation~\eqref{eq:Loss-Indicator-exp}, we get that
\begin{equation}\label{eq:Loss-Indicator-EXP-CS}
 \E[L] \le \sqrt{\E[F^2_{output}] \Pr(F_{output} > 2F_0)}.
\end{equation}
Recall from Theorem~\ref{thm:F_0BJKST} that $\E[F_{output}] = F_0$ and 
$\Var(F_{output}) = \frac{F_0(F_0 - k + 1)}{k -2}$.
It follows therefore that
\begin{align}\label{eq:second-moment}
      \begin{aligned}
               \E[F^2_{output}] &= (\E[F_{output}])^2 + \Var(F_{output})\\
                                &= F_0^2 + \frac{F_0(F_0 - k + 1)}{k -2}
                                = F_0^2 + \frac{F_0^2 -F_0(k-1)}{k -2}\\
                                & \le F_0^2 + \frac{F_0^2}{k-2} = F_0^2\left(1 + \frac{1}{k-2}\right).
             \end{aligned}    
\end{align}
On the other hand, the upper-tail bound of $\KMV$-estimator of Theorem~\ref{thm:F_0BJKST}
(for relative error $\varepsilon = 1$,
gives
\begin{equation}\label{eq:upperTailKMV}
      \Pr(F_{output} > 2F_0) \le 2e^{-ck},
\end{equation}
for some constant $c$.
Substituting~\eqref{eq:second-moment} and~\eqref{eq:upperTailKMV} into~\eqref{eq:Loss-Indicator-EXP-CS} gives:
\[
      \E[L] \le F_0 \sqrt{2\left( 1 + \frac{1}{k-2}\right)e^{-ck} }.
\]
For $k \ge 3$, 
$1 + \frac{1}{k -2} \le 2$,
and hence 
$\E[L] \le 2F_0 e^{-ck/2}$.
Combining with Equation~\eqref{eq:clipped-actual-loss} we have, 
\[
       \E[F_{clipped}] \ge F_0\left(1 - 2e^{-ck/2}\right).
\]
Since, clipping can only decrease the estimator, 
$\E[F_{clipped}] \le F_0$.
Thus we have 
\[
        F_0\left(1 - 2e^{-ck/2}\right) \le\E[F_{clipped}] \le F_0.
\]
Equivalently, we have
$\E[F_{clipped}] = F_0\left (1 - O(e^{-ck}) \right)$.\qedhere
\end{proof}

Next, we do a variance analysis of the clipped estimator.
\begin{lemma}[Variance of Clipped Estimator]\label{lem:clipped-estimator-varaince}
\[
      \Var(F_{clipped}) \le \Var(F_{output}) = \frac{d(d- k + 1)}{k-2}
\]
\end{lemma}
\begin{proof}[Proof of Lemma~\ref{lem:clipped-estimator-varaince}]
    Let $F'_{output}$ be an independent copy of $F_{output}$, and let
    $F'_{clipped} = \min \{F'_{output}, C \}$.
    For any two values $a, b$, 
$\lvert\min \{a, c\} - \min\{b ,c\} \rvert \le \lvert a-b \rvert$.
Therefore,
\[
        \lvert F_{clipped} - F'_{clipped}  \rvert \le \lvert F_{output} - F'_{output} \rvert.
\]
Squaring and taking expectations
$ \E[(F_{clipped} - F'_{clipped})^2] \le \E[(F_{output} - F'_{output})^2]$.

For independent, identically distributed random variables $Z$ and $Z'$, 
$\Var(Z) = \frac{1}{2}\E[(Z - Z')^2]$. Hence, 
$\Var(F_{clipped}) \le \Var(F_{output})$.

Using the variance guarantee of $\KMV$-estimator of Theorem~\ref{thm:F_0BJKST},
$\Var(F_{output}) = \frac{F_0(F_0 - k + 1)}{k -2}$,
we get that
$\Var(F_{clipped}) \le  \frac{F_0(F_0 - k + 1)}{k -2}$.\qedhere
\end{proof}

Next, we prove the tail bounds for the clipped estimaator.
\begin{lemma}[Tail-Bounds of Clipped Estimator]\label{lem:clipped-estimator-tailbounds}
For any $\varepsilon \in (0,1)$,
\[
    \Pr\left(\lvert F_{clipped} -F_0 \rvert \ge \varepsilon F_0 \right) \le \Pr\left(\lvert F_{output} - F_0 \rvert \ge \varepsilon F_0\right).
\]
As a result, if the $\KMV$-estimator satisfies
$\Pr\left( \lvert F_{output} - F_0 \rvert  \ge \varepsilon F_0 \right) \le 2e^{-ck \varepsilon^2}$,
then
\[
        \Pr\left( \lvert F_{clipped} - F_0 \rvert  \ge \varepsilon F_0 \right) \le 2e^{-ck \varepsilon^2},
\]
\end{lemma}
\begin{proof}[Proof of Lemma~\ref{lem:clipped-estimator-tailbounds}]
    Since $F_0 \le U$, $F_0 \le C = 2U$.
Consider any value of $F_{output}$. If $F_{output} \le C$, then
$F_{clipped} = F_{output}$.
So the two errors $\lvert F_{clipped} - F_0 \rvert $ and $\lvert F_{output} - F_0\rvert$
are equal.

If $F_{output} > C$, then $F_{clipped} = C$, since $F \le C$, clipping moves the estimate 
closer to the true value $F_0$, and therefore
$\lvert F_{clipped} - F_0 \rvert \le \lvert F_{output} - F_0 \rvert$.
Thus, pointwise, 
$\lvert F_{clipped} - F_0 \rvert \le \lvert F_{output} - F_0 \rvert$. 

Hence, whenever the clipped estimator has error at least $\varepsilon F_0$,
the $\KMV$-estimator must also have error at least $\varepsilon F_0$.
Therefore,
\[
      \Pr\left( \lvert F_{clipped} - F_0 \rvert  \ge \varepsilon F_0 \right) \le
      \Pr\left( \lvert F_{output} - F_0 \rvert  \ge \varepsilon F_0 \right).
\]
Applying the tail-bound of the $\KMV$-estimator gives the result.
\end{proof}

\begin{proof}[Proof of Theorem~\ref{thm:F_0BJKST Clipped Constructive}] Theorem~\ref{thm:F_0BJKST Clipped Constructive} follows from Theorem~\ref{thm:F_0BJKST} and Lemmas~\ref{lem:clipped-expectation}, \ref{lem:clipped-estimator-varaince} and \ref{lem:clipped-estimator-tailbounds}. 
    
\end{proof}

\end{document}